\documentclass[10pt,aps,aps,prx,twocolumn, superscriptaddress]{revtex4-2}

\usepackage[english]{babel}

\usepackage{amsmath}
\usepackage{amsthm}
\usepackage{amssymb}
\usepackage{dsfont}
\usepackage{mathtools}
\usepackage{bm}
\usepackage{tensor}
\usepackage{hyperref}
\usepackage{colortbl}
\usepackage{graphics}
\usepackage[dvipsnames,table]{xcolor}
\usepackage[justification=justified, format=plain, caption=false]{subfig}
\usepackage[indLines = true, noEnd = true]{algpseudocodex}
\usepackage{algorithm}

\usepackage{tikz}
\usetikzlibrary{quantikz2}
\definecolor{darkred}{rgb}{0.55, 0.0, 0.0}
\definecolor{thistle}{rgb}{0.85, 0.75, 0.85}

\newtheorem{theorem}{Theorem}[section]
\newtheorem*{theorem*}{Theorem} 
\newtheorem{example}[theorem]{Example}

\newtheorem{proposition}[theorem]{Proposition}
\newtheorem{lemma}[theorem]{Lemma}
\newtheorem{definition}[theorem]{Definition}

\newcommand{\avg}[1]{\langle #1 \rangle}
\newcommand{\expected}[1]{\mathbb{E}\left[ #1 \right]}
\newcommand{\belln}[1]{\mathcal{B}_{#1}}
\newcommand{\coldot}[1][black]{\text{\Large\textcolor{#1}{\ensuremath\bullet}}}

\newcommand{\egperm}[1]{#1 \in \mathcal{S}_n}

\newcommand{\egpartperm}[2]{#1 \in \mathcal{S}_{\partition{#2}}}

\newcommand{\evolve}[1]{\mathcal{U} #1 \mathcal{U}^\dagger}
\newcommand{\evolveH}[1]{\mathcal{U}^\dagger #1 \mathcal{U}}
\newcommand{\fock}[0]{\mathcal{F}}
\newcommand{\follows}{\, \Rightarrow \,}

\newcommand{\hilbert}[0]{\mathcal{H}}

\newcommand{\markov}[0]{\mathcal{M}}
\newcommand{\modsq}[1]{\left| #1 \right| ^2}
\newcommand{\norm}[1]{\| #1\|}
\newcommand{\outcome}[0]{\underline{s}}
\newcommand{\opperm}[1]{{\hat{P}_{#1}}}
\newcommand{\pluseq}{\mathrel{+}=}

\newcommand{\partition}[1]{(\underline{#1})}
\newcommand{\partket}[1]{\ket{\psi_{\partition{#1}}}}
\newcommand{\partbra}[1]{\bra{\psi_{\partition{#1}}}}
\newcommand{\permament}[1]{\mathrm{Perm} \!\left( #1 \right)}
\newcommand{\refeqn}[1]{Eq.~(\ref{#1})}
\newcommand{\scalar}[2]{\langle#1| #2 \rangle}
\newcommand{\symg}[1]{\mathcal{S}_{#1}}
\newcommand{\sympart}[1]{\mathcal{S}_{\partition{#1}}}
\newcommand{\trace}[1]{\mathrm{Tr} \!\left[ #1 \right]}

\newcommand{\var}[1]{\mathrm{Var}\!\left( #1 \right)}
\newcommand{\wpart}[1]{p_{\partition{\Lambda}}}

\begin{document}

\title{Incoherent behavior of partially distinguishable photons}
\author{Emilio Annoni}
\email{Corresponding author: emilio.annoni@quandela.com}
\affiliation{ Quandela SAS, 7 Rue Léonard de Vinci, 91300 Massy, France}
\affiliation{ Centre for Nanosciences and Nanotechnology, Université Paris-Saclay, UMR 9001,10 Boulevard Thomas Gobert, 91120, Palaiseau, France}
\author{Stephen C. Wein}

\affiliation{ Quandela SAS, 7 Rue Léonard de Vinci, 91300 Massy, France}

%%%%%%%%%%%%%%%%%%%%%%%%%%%%%%%%%%%%%%%%%%%%%%%%%%%%%%%%%%%%%%%%%%%%%%%%%%%%%%%%
\begin{abstract}

% Photon distinguishability serves as a fundamental metric for assessing the quality of quantum interference in photocounting experiments. In the context of Boson Sampling, it plays a crucial role in determining classical simulability and the potential for quantum advantage. 

Photon distinguishability is a key factor limiting quantum interference in photonic devices, directly impacting the performance of protocols such as Boson Sampling and photonic quantum computing. We present a basis-independent framework for analyzing multi-photon interference, identifying a necessary and sufficient condition under which distinguishability behaves as a stochastic error process. This condition enables any multi-photon state to be uniquely expressed as a classical mixture of partition states---discrete configurations representing different patterns of photon distinguishability. We introduce an experimentally implementable operation, analogous to Pauli twirling, that enforces this condition without compromising computational hardness. The resulting probability distribution over partition states defines the system's incoherent distinguishability spectrum, which we show can be fully characterized through a specific set of experiments. Building on this structure, we also demonstrate the existence of an error mitigation strategy. This framework clarifies key challenges in defining genuine multi-photon indistinguishability, links previous perspectives on partial distinguishability, and provides a rigorous foundation for robust photonic protocols.

\end{abstract}
%%%%%%%%%%%%%%%%%%%%%%%%%%%%%%%%%%%%%%%%%%%%%%%%%%%%%%%%%%%%%%%%%%%%%%%%%%%%%%%%
\maketitle

\section{Introduction}

    Photonic systems are a leading platform for quantum information processing \cite{walmsley_light_2023}, offering near decoherence-free evolution and scalable architectures when paired with reliable multi-photon state preparation. As described in the Knill-Laflamme-Milburn (KLM) scheme \cite{knill_scheme_2001} or by leveraging the Measurement Based Quantum Computing paradigm \cite{raussendorf_oneway_2001}, discrete-variable photonic computation can be carried out using only linear-optical interference and photodetection.
    Demonstrating high-quality multi-photon interference is a natural benchmarking step on the path toward these more general computing schemes.

    Predicting the outcome of a rudimentary photonic computation---specifically, the photocounting statistics at the output of an arbitrary linear interferometer---requires overcoming a fundamental challenge: the physics of photon interference conceals a hard combinatorial problem~\cite{scheel_permanents_2004, valiant_complexity_1979}.
    The formalization of the task of simulating this prototypical experiment---Boson Sampling~\cite{aaronson_computational_2011}---exposed how photon interference is actually tied to a highly plausible separation between the classical and quantum model of computation.
    However, a classical computer might still be able to efficiently reproduce the results of such a computation if imperfections manage to steer an experiment too far from the quantum regime~\cite{shchesnovich_sufficient_2014,renema_efficient_2018,garcia-patron_simulating_2019, renema_classical_2019}. Uncovering how imperfections affect the quality of photon interference is therefore crucial, not only to assess experimental implementations of Boson Sampling~\cite{tillmann_experimental_2013, spring_boson_2013, wang_boson_2019, maring_versatile_2024}, but also to support claims of photonic quantum advantage \cite{aaronson_complexitytheoretic_2017}.

    Photon distinguishability---the presence of differences in the internal degrees of freedom of photons---remains one of the most theoretically challenging imperfections to understand. While unwanted multi-photon emission from single-photon sources is already strongly suppressed in modern devices \cite{somaschi_nearoptimal_2016,hanschke_quantum_2018}, and optical loss, though substantial, is relatively simple to understand \cite{garcia-patron_simulating_2019} and can be mitigated with post-selection or other strategies \cite{mills_mitigating_2024}, distinguishability introduces a wide range of complex and often coherent interference effects \cite{tichy_fourphoton_2011,menssen_distinguishability_2017} that are far more difficult to characterize or suppress.

   % Photon interference is primarily impacted by three imperfections: unwanted multi-photon components in prepared single-photon states, optical losses in the interference or detection apparatus, and differences in the internal state of each photon, collectively known as ``photon distinguishability". State-of-the-art devices based on single-photon emitters already demonstrate very low multi-photon contributions \cite{somaschi_nearoptimal_2016,hanschke_quantum_2018}. In contrast, photon loss remains a big issue for near-term devices, but its impact on photon interference is relatively simple to understand \cite{garcia-patron_simulating_2019}. Moreover losses can, in principle, be overcome by investing additional classical resources (time), through either post-selection or more refined mitigation techniques \cite{mills_mitigating_2024}. On the other hand, the effects arising from photon distinguishability remain challenging to address due to their broad and complex range of behaviors \cite{tichy_fourphoton_2011,menssen_distinguishability_2017}.

    Accurately modeling a state of $n$ partially distinguishable photons requires up to $n!$ independent parameters \cite{shchesnovich_partial_2015, tichy_sampling_2015}, making the task quickly intractable. This complexity arises from both rich classical correlations between the internal states of individual photons \cite{spivak_generalized_2022} and genuinely coherent effects linked to internal “geometric” phases \cite{shchesnovich_collective_2018}.
    Such intricacies make it impractical to retain a complete description of distinguishability when analyzing advanced photonic protocols---such as all-photonic repeaters \cite{azuma_allphotonic_2015}, KLM-like schemes \cite{knill_scheme_2001, maring_versatile_2024}, or fusion networks \cite{bartolucci_fusionbased_2023}---which involve post-selection, adaptivity, and additional layers of noise.
    
    This challenge mirrors a well-known problem in circuit-based quantum computing, where fully coherent noise models prohibit large-scale simulation and complicate error mitigation \cite{cai_quantum_2023a}. In the context of fault-tolerant quantum computing, evaluating the performance of an error-correction circuit under general coherent gate errors is typically infeasible, and even modest coherent components can lead to significant deviations from approximate incoherent noise models for large codes \cite{greenbaum_modeling_2018}. As a result, both theoretical analyses and numerical simulations often rely on simplified noise models involving incoherent Pauli errors \cite{gidney_stim_2021}, and experimental strategies frequently aim to tailor physical noise toward this regime—for instance, through randomized compilation \cite{jain_improved_2023}.

    \begin{figure}[ht]
        \centering
        \includegraphics[width = .70\linewidth]{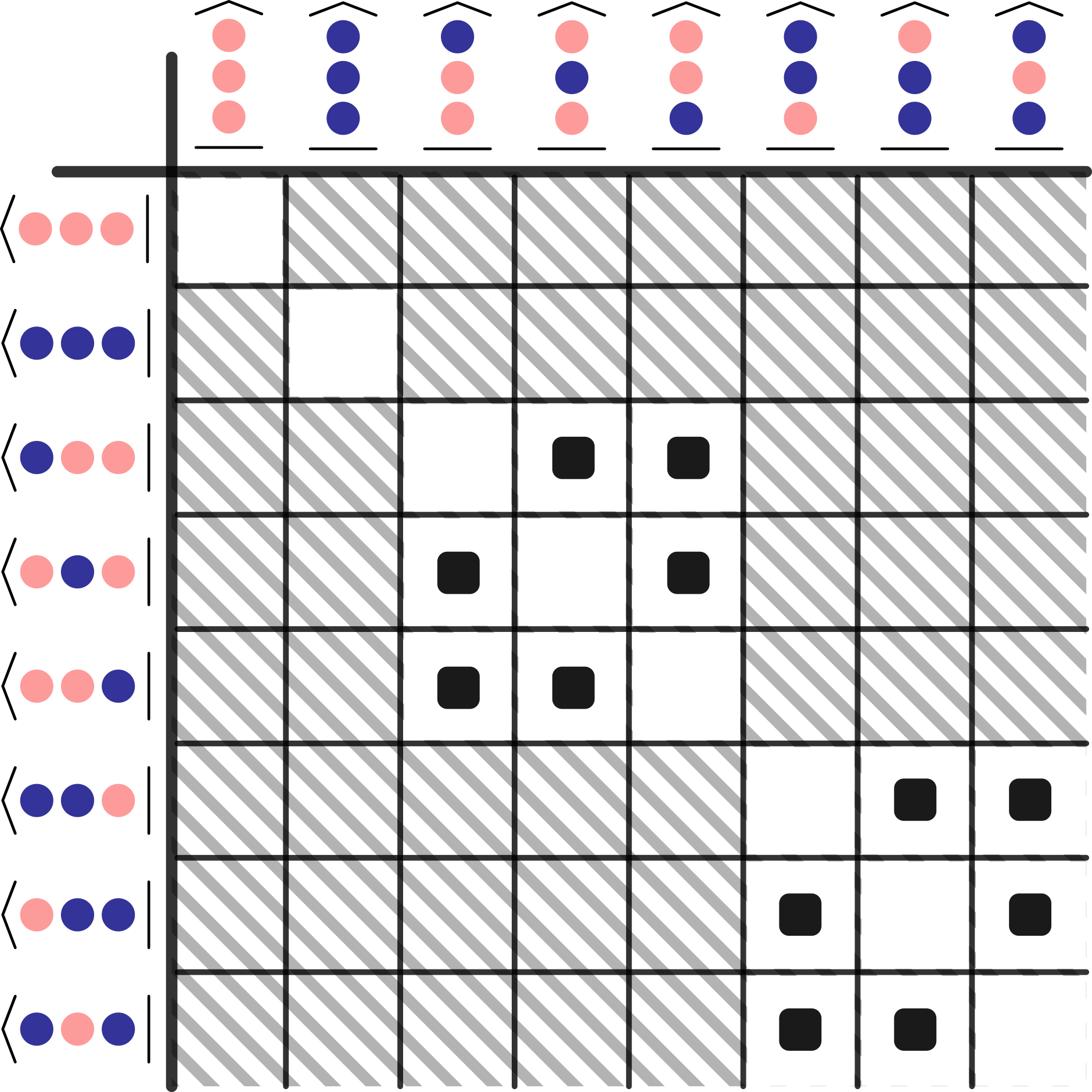}
        \caption{\textbf{An illustration of photocounting-relevant coherence.} Elements of the density matrix of a three-photon state $\bra{\coldot[gray]\coldot[gray]\coldot[gray]}\rho\ket{\coldot[gray]\coldot[gray]\coldot[gray]}$ with a two-dimensional internal state $\{\ket{\coldot[BlueViolet]}, \ket{\coldot[Apricot]}\}$. The black stamps depict the photocounting-relevant coherence, while the hatched area depicts the photocounting-irrelevant coherence---the share of the state whose effect cannot be detected by the photodetection observables. The diagram highlights the subspaces (in white) spanned by the orbit of each basis state under mode permutation.
        }
        \label{fig:irrelevant}

    \end{figure}

    % In the context of multi-photon interference the role of "discrete bit-flips" is played by complete distinguishability: switching the state of a photon so that it's orthogonal to the rest will make it interfere classically, producing a behavior which is easy to simulate \cite{aaronson_bosonsampling_2014}.
    % With this in mind, 
    The aim of this paper is to explore the notion of incoherent error in the context of multi-photon interference, by characterizing the regime in which correlations between the internal states of photons give rise to discrete, stochastic distinguishability errors. A formal understanding of this regime is a crucial first step toward developing simplified models of photonic noise—enabling more tractable analysis, efficient simulation techniques, practical error mitigation strategies, and targeted noise tailoring in quantum photonic systems.
    
    A first result of this paper is to show that the problem can be addressed in a basis-independent manner (Thm.~(\ref{th:state_equivM})). Inspired by the observation that ideal photodetectors resolve (but do not register) the internal state of each photon, we demonstrate that most coherences in an $n$-photon state are irrelevant for describing interference behavior and can be safely discarded without affecting the observed statistics (see Fig.~\ref{fig:irrelevant}).
    By linking the parameters of a known framework for partial distinguishability \cite{shchesnovich_partial_2015} to a concrete measurement procedure (App.~\ref{app:cyclic}), we show that all relevant features of a state can be accessed through a form of \emph{distinguishability tomography}.
    This approach bypasses the need for an explicit density-matrix description and avoids the ambiguities introduced by the implicit projection associated with photodetection.

    Using this tool, we derive the main result of this work: a characterization of all partially distinguishable states whose interference behavior can be attributed to discrete incoherent errors (Thm.~(\ref{th:partition})). This establishes a clear boundary between states belonging to an incoherent distinguishability regime (denoted by $\mathcal{I}$) and those exhibiting coherent distinguishability effects ($\overline{\mathcal{I}}$).
    
    In analogy with the qubit scenario, we show how incoherent distinguishability errors can be mitigated using probabilistic error cancellation (Thm.~\ref{th:mitigation}), and that any state can be brought to the incoherent distinguishability regime by applying a stochastic operation (Thm.~\ref{th:twirling}). Furthermore, we evidence that states transformed in this way retain computational hardness, supporting the argument that such operations do not diminish their usefulness for quantum information processing when applied appropriately.
    
    In Sec.~\ref{sec:simulation}, we analyze the structure of Boson Sampling in the presence incoherent distinguishability errors. In Sec.~\ref{sec:disambiguation} we discuss the notion of ``genuine $n$-photon indistinguishability"---a quantity that is central in the characterization of multi-photon distinguishability---and we demonstrate that it acquires a natural and consistent meaning within the incoherent distinguishability regime. However, we show that this interpretation is only partially compatible with other definitions found in the literature, including those based on representation-theoretic approaches \cite{tillmann_generalized_2015}. This mismatch unveils a previously overlooked fundamental ambiguity in the meaning of genuine indistinguishability, which our framework helps to clarify. We conclude in Section~\ref{sec:discussion} with a discussion and outlook.

%    When relevant we linked with other major approaches in the literature, such as the Orthogonal Bad Bits (OBB) mode \cite{sparrow_quantum_2017,shaw_errors_2023} and the representation theory framework . 

%%%%%%%%%%%%%%%
%-------------------------------------------------

\section{Permutation representation}
\label{sec:preliminaries}

A multi-photon interference experiment considers a physical system where a state of light occupying $n$ distinct input spatial modes evolves into $m$ output modes through the action of a linear interferometer. The input Hilbert space is a tensor product of $n$ identical Fock spaces $\hilbert = \bigotimes_i \fock_i$, and the basis of the space $\fock_i$ is represented using the bosonic creation operator $a_i^\dagger$ defining the Fock states $\ket{n}_i = \frac{1}{\sqrt{n!}} a_i^{\dagger n} \ket{0}_i$, respectively. The interferometer, represented by the single-shot evolution operator $\mathcal{U}$, acts linearly on the input creation operators such that, in Heisenberg picture:
 \begin{equation}
    \label{eq:BS_linear_evo}
    \evolveH{b^{\dagger}_{j}} = \sum_{j=1}^m U^*_{ij}a^{\dagger}_{i},
\end{equation}
for some $n \times m$ unitary matrix $U$, called the scattering matrix, and where $b_i^\dagger$ are bosonic creation operations for the output modes of the interferometer.

Ideally, this system is prepared in a standard initial state $\ket{1^n} = \bigotimes_i a_i^\dagger \ket{0}$ with one photon per input mode. Through the mode-mixing action of the interferometer, this state evolves into a non-separable superposition of states with different mode occupation in the output Hilbert space. Using the mode-occupation vector  $\outcome = [s_0, s_1, ...s_m]$ to label each output configuration, where the integer $s_i$ counts the number of photons in mode $i$, it becomes possible to construct the complete set of photocounting observables through the projection operators $\hat{O}^{\outcome} = \bigotimes_i \ket{s_i}_i\bra{s_i}_i$, such that the measured photocounting probabilities are:
 \begin{equation}
 p(\outcome | 1^n) = \trace{\evolve{\ket{1^n}\bra{1^n}} \hat{O}^{\outcome}}.
 \end{equation}
 In this context, Boson Sampling is the problem of devising an efficient classical procedure which, for any given scattering matrix $U$, produces outcome samples $\outcome$ according to the distribution $p(\outcome | 1^n)$.
 
To incorporate photon distinguishability, we can consider a more realistic definition of $\fock_i$ that contemplates a basis of orthogonal internal degrees of freedom $a^\dagger_{i, \alpha}$ with associated Fock states $\ket{n_\alpha}_i = \frac{1}{\sqrt{n!}} a_{i,\alpha}^{\dagger n} \ket{0}$.
It is then useful to define the $k$-particle subspace $\fock_i^{(k)}$ of $\fock_i$, spanned by the vectors $\{\bigotimes_\alpha \frac{1}{\sqrt{j_\alpha!}} a_{i,\alpha}^{\dagger j_\alpha}\ket{0}_i\}$ respecting, under proper symmetrization, $\sum_\alpha j_\alpha = k$. In this more general version of Boson Sampling, the initial state is an arbitrary state $\ket{\psi} \in \bigotimes_{i=1}^n \fock_i^{(1)}$.
 The linear interferometer is assumed to act equally on all internal degrees of freedom, such that \refeqn{eq:BS_linear_evo} continues to be valid after appending the subscript $\alpha$ to creation operators on both sides.
 In the same way, we obtain a new set of observables $\hat{O}^{\outcome}$ by substituting each $\ket{s_i}_i\bra{s_i}_i$ with the projection operator over $\fock_i^{(s_i)}$.
 In this broader context, the task is to understand how the choice of initial state $\ket{\psi}$ affects the observed statistics $p(\outcome|\psi)$, and how the complexity of the corresponding sampling problem compares to the ideal scenario using indistinguishable particles.

 We can obtain a clear picture of the effects of partial distinguishability by expressing the photocounting probabilities $p(\outcome)$ as arising from the interference of virtual paths through the interferometer \cite{shchesnovich_partial_2015} (see Fig.~\ref{fig:path}).

  \begin{figure}[t]
    \includegraphics[height = .8\linewidth]{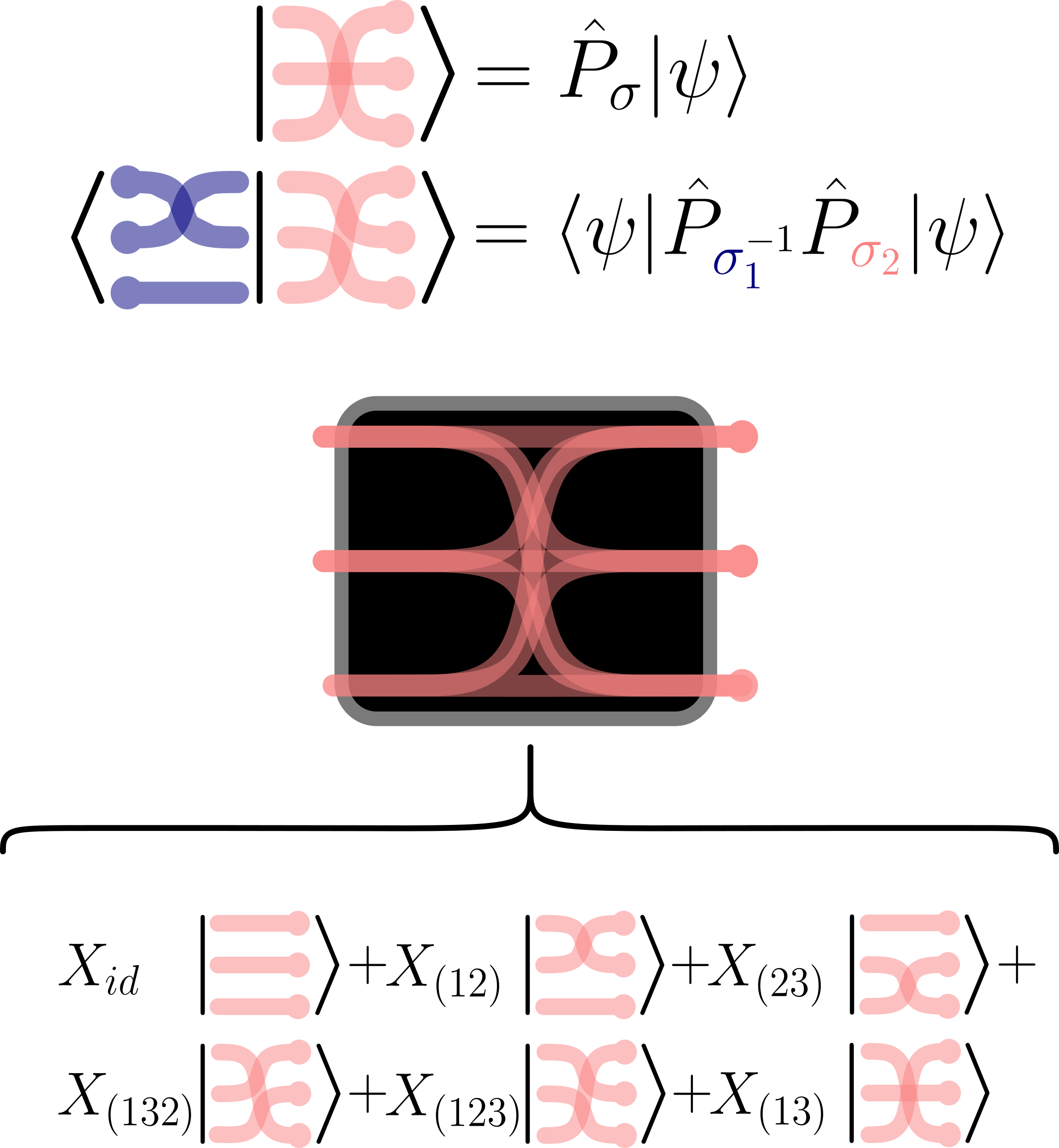}
    \caption{
    \textbf{A graphical representation of the decomposition of photon interference into paths through an interferometer.} The three-photon outcome $\outcome = [1,1,1]$ arising from a $3 \times 3$ interferometer is described by the coherent interference of all the $3!=6$ permutations path that could have led to its realization. The probability amplitude $X_\sigma$, indexed by the permutation $\sigma$ that describes a path, is given by the scattering matrix of the interferometer $U$ (e.g. $X_{(132)} = U_{13}U_{32}U_{21}$).
    The scalar products between two paths taken by the photons correspond to one of the indistinguishabilities of Shchesnovich's framework~\cite{shchesnovich_partial_2015}, here denoted as $M_\sigma$.
    \label{fig:path}
    }
\end{figure}

 If we pick an arbitrary element of the $n$-photon basis with internal states $\alpha_i$ and pass it through the interferometer we obtain:
 \begin{equation}
 \label{eq:allevolve}
     \mathcal{U}\prod_i a^\dagger_{i,\alpha_i} \ket{0} = \prod_i \left(\sum_j U_{ij}b^\dagger_{j,\alpha_i}\right) \ket{0}.
 \end{equation}
 By expanding the left hand side we recover $m^n$ distinct terms which can be listed through all the possible functions $f: \{1,..,n\}  \rightarrow\{1,...,m\}$. Each term corresponds to a discrete mode hopping event, where the photon in the input mode $i$ and internal state $\alpha_i$ is found in the output mode $f(i)$ with an amplitude of probability equal to:

 \begin{equation}
    \label{eq:path}
    X_{f} = \prod_{i=1}^n U_{i \,f(i)}.
 \end{equation}
  
 Using the projector $\hat{O}^{\outcome}$ we can further sort the paths based on the outcome they produce.
 Denoting with $\outcome_f$ the outcome linked to $f$, such that $(\outcome_f)_j = |f^{-1}(j)| $,
 and applying $\hat{O}^{\outcome}$ on \refeqn{eq:allevolve}, we're able to filter on the condition $\outcome_f = \outcome$ obtaining the relevant set of functions $\{f^{[\outcome]}\}$:

  \begin{equation}
    \label{eq:allpath}
    \hat{O}^{\outcome}\mathcal{U}\prod_i a^\dagger_{i,\alpha_i} \ket{0} = \sum_{f \in \{f^{[\outcome]}\}} \prod_i U_{i f(i)} b^{\dagger}_{f(i)}.
 \end{equation}

 Considering without loss of generality the scenario where $\outcome$ has no bunching and the photons are found in the first $n$ output modes ( $\outcome = [1,...,1,0,...,0]$ ) the set $f^{[\outcome]}$ obtained after applying  $\hat{O}^{\outcome}$ consists in all bijective functions from $\{1,...,n\}$ onto itself, that is the symmetric group $\symg{n}$ of all possible permutations $\sigma$ of $n$ objects.
 Upon the substitution $b^\dagger_{i,\alpha} \leftrightarrow a_{i,\alpha}^\dagger$, we can then rewrite the share of the state that produces $\outcome$ in terms of mode permutation operators $\opperm{\sigma}$, which act on $a^\dagger_i$ as $\opperm{\sigma} a_i^\dagger = a_{\sigma(i)}^\dagger$ (irrespective of the internal state).
 Using $X_\sigma$ to address the amplitude of probability that produces $\outcome$ through the mode hopping $\opperm{\sigma}$, we obtain the following equality, valid for an arbitrary basis element and thus valid for an arbitrary $\ket{\psi} \in \bigotimes_{i=1}^n \fock_i^{(1)}$:
 
  \begin{equation}
 \label{eq:evolution_permutation}
     \hat{O}^{\outcome} \mathcal{U}\ket{\psi} = \sum_{\sigma \in \symg{n}} X_\sigma \opperm{\sigma} \ket{\psi}.
 \end{equation}

 This last expression shows how the share of the output state pertaining to the subspace with mode occupation $\outcome$ can be seen as a superposition of all the possible permuted versions of the initial state $\ket{\psi}$ weighted by $X_\sigma$.
 For a state made of $n$ completely indistinguishable photons, the permutation operators act trivially on $\ket{\psi}$ leading to just one distinct pure state at the output. However, here we directly see how the coherent interference between the $n!$ paths is sensitive to the internal state of the photons. More explicitly computing $p(\outcome) = \bra{\psi} \mathcal{U}^\dagger\hat{O}^{\outcome} \mathcal{U}\ket{\psi}$ using \refeqn{eq:evolution_permutation}, it is possible to highlight the cross-path interference in a pairwise fashion:
 
 \begin{equation}
    \label{eq:path_interf}
    p(\outcome) = \sum_{\egperm{\sigma, \sigma'}} X^*_{\sigma'}X_\sigma \bra{\psi}\opperm{\sigma'^{-1}\sigma} \ket{\psi},
\end{equation}

\noindent where we have composed the action of the two permutation operators, noting that $\opperm{\sigma}^\dagger = \opperm{\sigma^{-1}}$.
\refeqn{eq:path_interf} shows that, in order to capture the effects of partial distinguishability, one must know the value of $n!$ independent quantities that govern the strength of coherent path interference. These are exactly computed by the (possibly complex) average value of the $n!$ mode permutation operators:

\begin{equation}
    M_\sigma = \langle \opperm{\sigma} \rangle.
\end{equation}

The quantities $M_\sigma$, here referred to as ``generalized indistinguishabilities", are a characterization of the symmetries of the $n$-photon initial state upon mode permutation, and represent all information about distinguishability that an ideal photocounting experiment is able to resolve.
This is the essence of Shchesnovich's distinguishability framework \cite{shchesnovich_partial_2015}, where this new set of parameters is shown to generalize to mixed initial states and all outcomes $\outcome$.

Importantly, in the ideal case where $\opperm{\sigma^{-1}\sigma'} \ket{1^n} = \ket{1^n}$ we have $M_\sigma = 1$ for all $\sigma$ and one can recover the well known permanent formula:
\begin{equation}
    p(\outcome | 1^n) = \left(\sum_{\egperm{\sigma}} X_\sigma\right)^*\left(\sum_{\egperm{\sigma}} X_\sigma\right) = \modsq{\permament{U^{[\outcome]}}} .
\end{equation}

This description of distinguishability exposes how, in general, a complete tomography of a multi-photon initial is redundant if the only goal is to describe the photocounting statistics it produces at the output of a linear interferometer. No matter the dimension of the space that is required to describe the state, it is sufficient to determine only $n!$ parameters. Considering all, possibly-mixed, partially-distinguishable $n$-photon states, notated here by the set $\mathcal{D}$ of density matrices over $\bigotimes_{i=1}^n \fock_i^{(1)}$, it is useful to establish an equivalence relation that dictates whether two such states can be told apart using only linear optics and photocounting:

\begin{definition}
\label{def:state_equiv}
    Given $\rho$, $\rho'$ $\in \mathcal{D}$, $\rho \sim \rho'$ if and only if $p(\outcome | \rho) = p(\outcome | \rho')$ for all unitaries $U$ and outcomes $\outcome$.
\end{definition}

A relevant feature of $\sim$ is that, being based on the observables $\hat{O}^{\outcome}$, it obeys the same algebra of expectation values and is preserved upon mixing:

\begin{equation}
    \rho_1 \sim \rho'_1,\rho_2 \sim \rho'_2 \Rightarrow \rho_1 + \rho_2 \sim \rho_1' + \rho_2'.
\end{equation}
Given \refeqn{eq:path_interf}, if two states present the same generalized indistinguishabilities $M_\sigma$, they will produce the same photocounting statistics in every interference experiment. In addition Lemma \ref{lm:shch_sufficient}, proven in Appendix~\ref{app:cyclic}, shows that the converse is also true: each generalized indistinguishability can be unambiguously measured with a carefully chosen photocounting experiment that generalizes the cyclic interferometer proposed in Ref.~\cite{pont_quantifying_2022}. This family of $2n$-mode interferometers is constructed by applying a layer of $n$ balanced beam splitters, followed by one or more phase shifters, a permutation applied to every second mode, and a final layer of balanced beam splitters (see Fig.~\ref{fig:cyclic} in Appendix~\ref{app:cyclic}). The interferometer described in Ref.~\cite{pont_quantifying_2022} can thus be understood as the specific case where the permutation consists of a single $n$-cycle.

This explicit construction allows Def.~\ref{def:state_equiv} to be equivalently formulated in terms of the generalized indistinguishabilities.
Considering two $n$-photon states $\rho, \rho'$ and their indistinguishabilities $M_\sigma = \langle \hat{P}_\sigma \rangle_\rho$ and $M'_\sigma = \langle \hat{P}_\sigma\rangle_{\rho'}$ we have the following:
\begin{theorem}
\label{th:state_equivM}
    Given $\rho$, $\rho'$ $\in \mathcal{D}$, $\rho \sim \rho' \iff M_\sigma = M'_\sigma$ for all $\sigma \in \symg{n}$,
\end{theorem}
\begin{proof}
    The statement follows from Lemmas \ref{lm:shch_necessary} and \ref{lm:shch_sufficient}.
\end{proof}

We thus prove that this characterization of distinguishability, formulated in terms of permutation observables, constitutes a non-redundant and basis-independent quorum of information sufficient to describe the outcome of any interference experiment performed on the state.

%%%%%%%%%%%%%%%%%%%%%%%%%%%%%%%%%%%%%%%%%%%%%%%%%%%%%%%%%%%%%%%%%%%%%%%%%%%%%%%%
\section{Partition representation}
\label{sec:reconstruct}

In the previous section, we explored the complete description of multi-photon interference in terms of path permutations, which captures both coherent and incoherent effects. In this section, we focus on deriving the restricted set of states where imperfections due to partial distinguishability give rise to purely incoherent behavior.

    We begin by fixing a basis $\ket{1_\alpha}$ for all the single particle subspaces $\fock^{(1)}$, and consider the most general $n$-photon state $\rho \in \mathcal{D}$ with a discrete distinguishability structure:

    \begin{equation}
    \label{eq:general_decomp}
        \rho = \!\!\sum_{\alpha_1, \hdots, \alpha_n} p(\alpha_1,\hdots, \alpha_n) \ket{1_{\alpha_1}\hdots1_{\alpha_n}}\bra{1_{\alpha_1}\hdots1_{\alpha_n}},
    \end{equation}
    \noindent where the sum runs over all possible distinct internal states $\alpha_i$ for the $n$ particles.
    States with such structure can exhibit classical correlation but possess no coherence between the internal state of each photon.

\begin{figure}[htb]
    \centering
    \includegraphics[width = .65\linewidth]{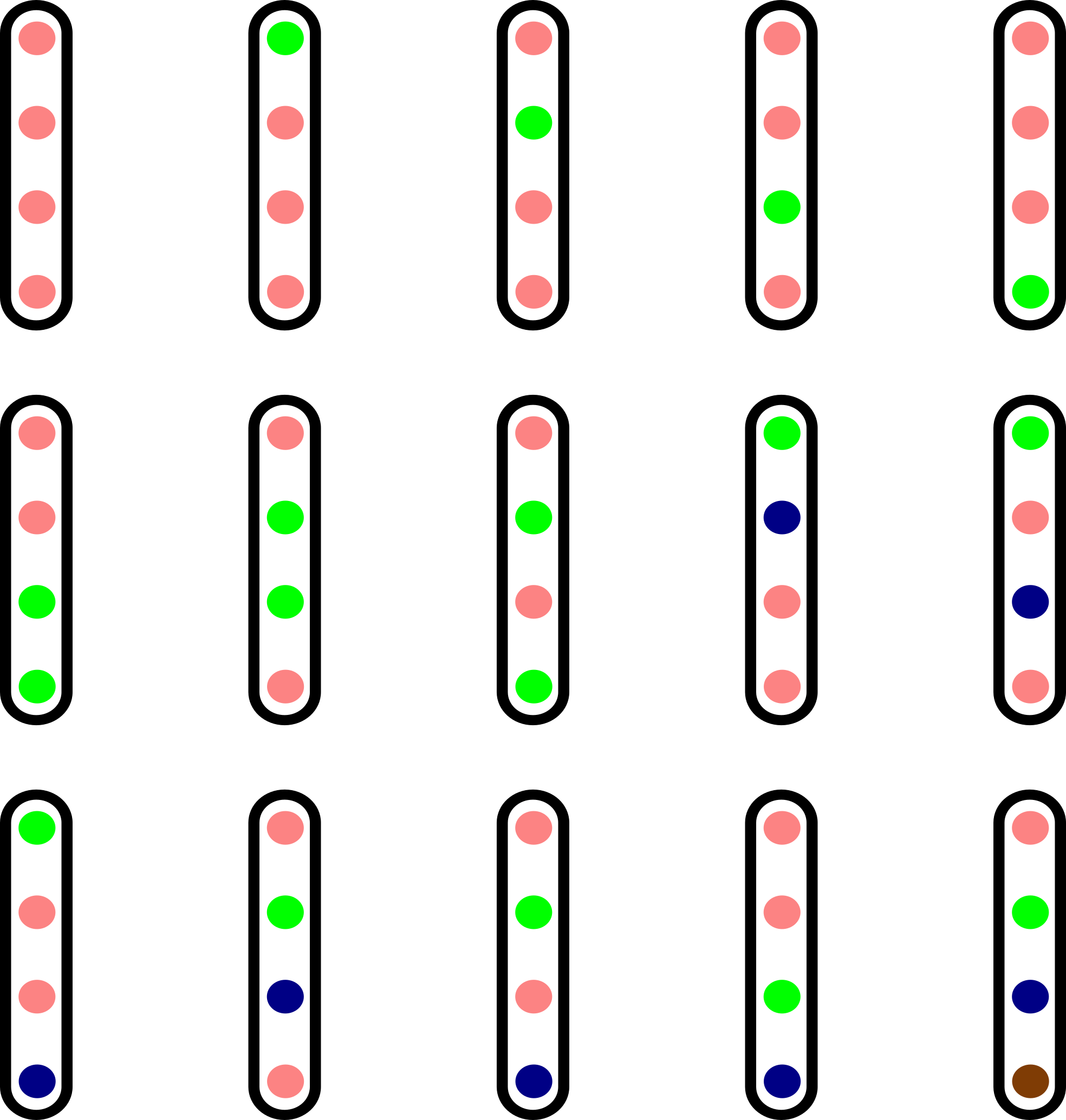}
    \caption{ \textbf{An illustration of partition states} All 15 possible discrete distinguishability configurations of four photons in four modes. In the diagram, photons of the same color are fully indistinguishable ($\scalar{\coldot[Apricot]}{\coldot[Apricot]} = 1$),  while photons of different colors are fully distinguishable ($\scalar{\coldot[Apricot]}{\coldot[BlueViolet]} = 0$).}
    \label{fig:configurations}
\end{figure}
    When focusing on a single eigenstate of this decomposition, it becomes evident that its behavior upon permutation, and thus the value of its generalized indistinguishabilities $M_\sigma = \langle \opperm{\sigma}\rangle$, depends only on the state's \emph{configuration} of internal states and not on the specific basis elements that realize it. For example: $ \ket{1_a 1_b 1_b 1_c 1_a} \sim \ket{1_d 1_a 1_a 1_f 1_d}$. These distinguishability configurations are illustrated in Fig.~\ref{fig:configurations} for $n=4$ particles.
    The various groupings of indistinguishable photons among these particles induce a ``distinguishability partition" $\partition{\Lambda}$. Formally, this is a collection of non-intersecting subsets of the modes $\Lambda_i$ that sum to the whole set. Using Thm.~\ref{th:state_equivM}, it is then possible to group together all states that exhibit the same behavior under photocounting, globally obtaining a single probability distribution $\{ \wpart{\Lambda}\}$ for all such configurations.
    By choosing an arbitrary set of representatives for each distinguishability partition $\{ \partket{\Lambda} \}$, we can rewrite all states of the form of \refeqn{eq:general_decomp} as:
    
    \begin{equation}
        \label{eq:part_repr}
     \rho \sim \sum_{\partition{\Lambda}} \wpart{\Lambda} \partket{\Lambda} \partbra{\Lambda}.
    \end{equation}
    We then define the class $\mathcal{I} \subset \mathcal{D}$ of states which exhibit a incoherent distinguishability behavior as the ones for which the decomposition of \refeqn{eq:part_repr} is possible:

    \begin{definition}[Incoherent distinguishability]
        Let $\rho$ be a partially distinguishable $n$-photon state:
        $$ \rho \in \mathcal{I} \iff \;\exists\: p_{\partition{\Lambda}} \text{ such that } \rho \sim \sum_{\partition{\Lambda}} \wpart{\Lambda} \partket{\Lambda} \partbra{\Lambda} .$$
    \end{definition}

    The task of characterizing $\mathcal{I}$  can thus be seen as an instance of a quasi-probability reconstruction problem \cite{sperling_quasiprobability_2018}: we are asked to recover the features of a given initial state $\rho$ through a possibly positive mixture of states with a ``classical" behavior.
    In our scenario, the reference set of states defining the classical behavior of distinguishability are the distinguishability configurations illustrated in Fig.~\ref{fig:configurations}---the partition states $\partket{\Lambda}$.
    The added subtlety of our problem is that, in contrast to Ref.~\cite{sperling_quasiprobability_2018}, our knowledge of the state is only indirect and passes through the eyes of the observables $M_\sigma$.

    In this context, the existence of a partition representation (and its positivity) can be considered a classicality boundary for the behavior of distinguishability (see Fig.~\ref{fig:classicality}), adding further structure on top of the known difference between pure \cite{tichy_sampling_2015} and mixed \cite{shchesnovich_partial_2015} states.
    Furthermore an interesting feature of such a representation is that the behavior of each partition state $\partket{\Lambda}$ can be expressed in terms of the properties of the associated partition $\partition{\Lambda}$.
    In particular, partitions can be partially ordered by a set inclusion criterion ($\succeq$, see Appendix~\ref{app:partition}) and the indistinguishabilities $M_\sigma$ of each $\partket{\Lambda}$ can be obtained by comparing $\partition{\Lambda}$ with the partition $\partition{\sigma}$ induced by the disjoint cycles of $\sigma$:

    \begin{equation}
        \label{eq:partition_msigma}
        M_\sigma = 
        \begin{cases}
        1 & \text{if } \partition{\Lambda} \succeq \partition{\sigma} \\
        0 & \text{otherwise}
        \end{cases}.
    \end{equation}

    This last equation captures the intuitive statement that, in order to obtain non-zero indistinguishability, $\sigma$ cannot move photons between two distinct cells $\Lambda_i \neq \Lambda_j$ of the partition $\partition{\Lambda}$, which contain photons in orthogonal states.
    We can then observe how, for a partition state $\partket{\Lambda}$, the values of $M_\sigma$ no longer explicitly depend on the permutation $\sigma$, but rather on its cycle partition $\partition{\sigma}$.
    This property, which we refer to as ``orbit invariance", is clearly preserved by mixing in \refeqn{eq:part_repr}, and is thus a necessary condition for a state to belong to $\mathcal{I}$:

    \begin{lemma}
    \label{lm:partition_necessary}
        Let $\rho\in \mathcal{I}$ with distinguishabilities $M_\sigma$ and $\sigma, \tau \in \mathcal{S}_n$:
        $$\partition{\sigma} = \partition{\tau} \follows M_\sigma = M_\tau .$$
    \end{lemma}

    A first consequence of constraining distinguishability to an effect of classical correlations is that the original set of $n!$ independent parameters $M_\sigma$ must be reduced to a size equal to the number of distinct partitions $\belln{n}$ called the Bell number (see Appendix~\ref{app:partition}). Moreover, orbit invariance is actually enough to recover a partition distribution that can reproduce the interference behavior of the state:
    
    \begin{lemma}
    \label{lm:partition_sufficient}
        Let $\rho\in\mathcal{D}$ with distinguishabilities $M_\sigma$, such that $\partition{\sigma} = \partition{\tau} \follows M_\sigma = M_\tau$.
        Then there exists a $ \rho' \in \mathcal{I} $ such that $\rho \sim \rho'$.
    \end{lemma}

    The proof of Lm.~\ref{lm:partition_sufficient} relies on inverting a system of equations based on the order relation $\succeq$ between partitions.
    Problems of this kind are already known in the literature and are linked to the same combinatorial relations that leads to the ``inclusion-exclusion" principle exploited in Ryser's algorithm~\cite{ryser_combinatorial_1963}, namely M\"obius inversion.
    We thus rely on the following result: 

    \begin{lemma}[M\"obius inversion \cite{rota_foundations_1964}]
    \label{lm:mobius}
        Let $\succeq$ be a partial order relation over a set $\Omega$.
        The  $|\Omega|\times |\Omega|$ matrix 
        $$
        R_{ab} =
        \begin{cases}
            1 & \text{if } a \succeq b\\
            0 & \text{otherwise}
        \end{cases}
        $$
        is invertible.
    \end{lemma}
    
    \begin{proof}[Proof of Lm.~\ref{lm:partition_sufficient}]
    If we consider a state $\rho'\in \mathcal{I}$, with a partition distribution $\wpart{\Lambda}$, we can use \refeqn{eq:partition_msigma} to first express the generalized indistinguishabilities of each partition state. By summing over all possible partitions, the generalized indistinguishabilities $M'_\sigma$ of $\rho'$ are then:
    \begin{equation}
    \label{eq:msigma_sum1}
        M'_\sigma = \sum_{\partition{\Lambda} \succeq \partition{\sigma}} \wpart{\Lambda}.
    \end{equation}
    
    Taking into account the resizing implied by orbit invariance, \refeqn{eq:msigma_sum1} defines a system of $\belln{n}$ distinct linear equations that compute the generalized indistinguishabilities $M'_\sigma$ given a partition distribution $\wpart{\Lambda}$.
    Rewriting this system in matrix form, one obtains
    $ M'_\sigma~=~\sum_{\partition{\Lambda}} R_{\partition{\sigma}\partition{\Lambda}}\cdot\wpart{\Lambda}$
    with : 
    $$
    R_{\partition{\sigma}\partition{\Lambda}} =
        \begin{cases}
            1 & \text{if } \partition{\sigma} \preceq \partition{\Lambda}\\
            0 & \text{otherwise}
        \end{cases}.
    $$
    According to Lm.~\ref{lm:mobius}, $R_{\partition{\sigma}\partition{\Lambda}}$ can be inverted, recovering a partition distribution $\wpart{\Lambda}$ which reproduces the distinguishabilities $M_\sigma$ of $\rho$.
    
    \end{proof}

    For an interested reader, we also provide a more elementary proof of Lm.~\ref{lm:partition_sufficient} in Appendix~\ref{lm:perm_to_part} that is based on the determinant.
    
    By combining Lemmas \ref{lm:partition_sufficient} and \ref{lm:partition_necessary}, we can finally obtain the desired characterization of the class $\mathcal{I}$:
    \begin{theorem}
    \label{th:partition}
        Let $\rho\in \mathcal{D}$ with distinguishabilities $M_\sigma$.
        Then $\rho\in \mathcal{I}$ if and only if $\partition{\sigma} = \partition{\tau} \follows M_\sigma = M_\tau$.
    \end{theorem}

    {\begin{figure}[t]
        \includegraphics[width = .8\linewidth]{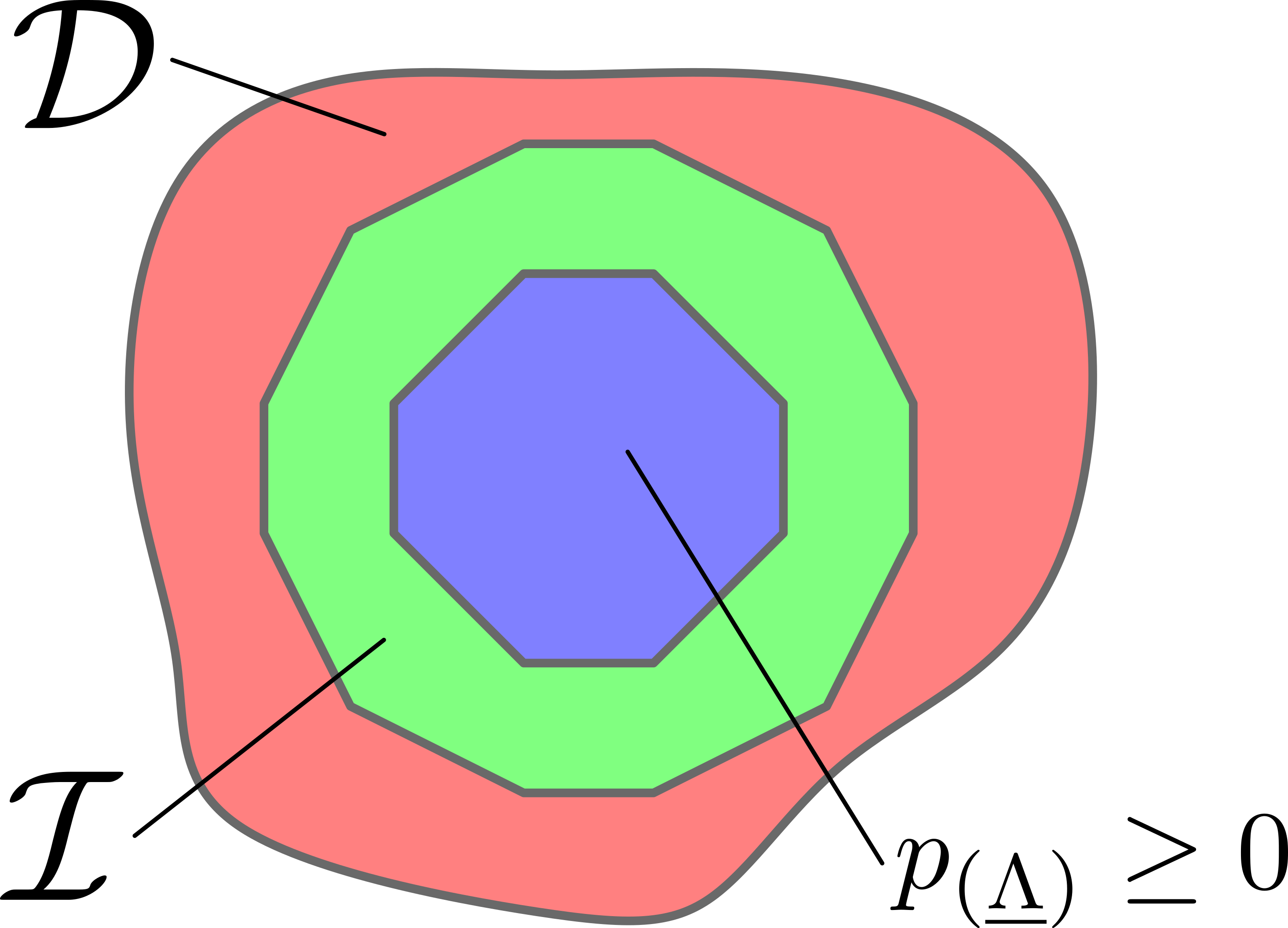}
        \caption{\textbf{Illustration of the regimes of partial distinguishability drawn from partition arguments.} The existence of a partition representation marks a first boundary between arbitrary distinguishability behavior (red area) and general incoherent distinguishability (green area). A second classicality boundary is crossed when the partition representation is positive (blue area).}
        \label{fig:classicality}
    \end{figure}

    A relevant situation in which states of the form given in Eq.~(\ref{eq:general_decomp}) naturally arise is when the $n$ single-photon states composing $\rho$ are simultaneously diagonalizable. By analyzing the generalized indistinguishability parameters of an arbitrary separable $n$-photon state, we uncover a direct connection between simultaneous diagonalizability and the condition established in Theorem~\ref{th:partition}:

    \begin{proposition}[Separable states \cite{minke_characterizing_2021}]
    \label{eq:msigma_separable}
        Let $\rho\in\mathcal{D}, \quad\rho= \bigotimes_i^n \rho_i$ and let us express $\sigma\in \mathcal{S}_n$ as the product of its disjoint cycles $\sigma = \prod_i\sigma_i$:

        $$
        M_\sigma =\avg{\opperm{\sigma}}_\rho = \prod_{\sigma_i} \trace{\vec{\prod_{j \in \sigma_i}} \rho_j}.
        $$
        
    \end{proposition}

    \noindent That is, for each cycle $\sigma_i$, one must compute the trace of the ordered product of the single-photon density matrices corresponding to the input modes involved in the cycle. The arrow over the second product symbol emphasizes that the product is non-commutative---the density matrices must be multiplied in the specific order dictated by the cycle. Requiring orbit invariance means that these traces must remain unchanged under any reordering of the matrices  (e.g. $\trace{\rho_1\rho_2\rho_3\rho_4} = \trace{\rho_1\rho_3\rho_4\rho_2}$). In other words, orbit invariance requires these operators to effectively commute within the trace, which---unlike a conventional commutator---discards coherence contributions that are irrelevant to photocounting statistics.

    The fact that the partition distribution $\wpart{\Lambda}$ of a state $\rho\in \mathcal{I}$ is uniquely defined by the value of the observables $M_\sigma = \langle \opperm{\sigma}\rangle$, means that it is indeed a property of the state with a well-defined, basis-independent meaning.
    Furthermore whenever a partition representation can be defined, the equivalence criterion of Thm.~\ref{th:state_equivM} can be extended:
    \begin{theorem}[Partition equivalence]
    \label{th:equiv_part}
    Let $\rho, \rho' \in \mathcal{I}$:
        $$\rho \sim \rho' \iff M_{\sigma} = M'_{\sigma} \iff \wpart{\Lambda} = \wpart{\Lambda}'.$$
    \end{theorem}
    \begin{proof}
        According to  Lm.~\ref{lm:partition_sufficient}, $R_{\partition{\sigma}\partition{\Lambda}}$ is invertible and thus realizes a bijection between orbit invariant $M_\sigma$ and the probability distributions $\wpart{\Lambda}$.
    \end{proof}

    In essence, no two different partition distributions $\{ \wpart{\Lambda}\}$ and $\{ \wpart{\Lambda}' \}$ can produce photocounting-equivalent states $\rho \sim \rho'$ as this would imply some 0 eigenvalue for $R_{\partition{\sigma}\partition{\Lambda}}$. 
    This result shows that the various behavior of states in the class $\mathcal{I}$ are exhaustively listed by all the distinct (quasi)probability distributions $ \{ \wpart{\Lambda} \}$.
    
    A selection of examples that illustrate relevant states in $\mathcal{I}$ and $\overline{\mathcal{I}}$ is given in Appendix~\ref{app:examples}.
    
%%%%%%%%%%%%%%%%%%%%%%%%%%%%%%%%%%%%%%%%%%%%%%%%%%%%%%%%%%%%%%%%%%%%%%%%%%%%%%%%

    \section{Features of incoherent distinguishability}

We now examine key features of the incoherent distinguishability regime that highlight its utility in near-term photonic applications. We show that incoherent distinguishability errors can be mitigated through experimentally feasible partitioning strategies, that this regime can always be recovered using a permutation twirling operation, and that states within this regime retain computational hardness for tasks such as Boson Sampling.

   \subsection{Incoherent distinguishability can be mitigated}

    A useful property of the path permutation description of multi-photon interference is that the indistinguishabilities $M_\sigma$ behave well under particle subtraction: the mode permutation operator $\opperm{\sigma}$ traces out any particle that is not moved by $\sigma$.
    From the set of $n!$ indistinguishabilities $M_\sigma$, one can thus readily extract information that pertains only to a subset of the original $n$ particles.
%    Equivalently, when the distinguishabilities are orbit invariant, we can obtain the partition distribution of a subset of the $n$ particles by either tracing out the partition distribution of the original state or by inverting the permutation to partition relation on the relevant subset of distinguishabilities.

    Knowing this, we can exploit a ``partitioning" operation---similar to that proposed in \cite{tichy_fourphoton_2011}---which can be experimentally implemented by delaying photons into well-separated temporal groups.
    Given a target partition $\partition{\Lambda}$, the photons within each cell $\Lambda_i$ are delayed together until they become perfectly distinguishable from photons belonging to other cells, resulting in a state that we will denote $\rho\!\!\downharpoonright_{\partition{\Lambda}}$ (see Fig.~\ref{fig:mitigation}).
    This operation enforces a useful condition on $M_\sigma$  which, exactly as in the case of the partition state $\partket{\Lambda}$, will be different from 0 only if $\sigma$ permutes photons within the cells of $\partition{\Lambda}$, so that:
    \begin{equation}
        \label{eq:partitioning_condition}
         M'_\sigma =         \begin{cases}
            M_\sigma & \text{if } \partition{\sigma} \preceq\partition{\Lambda}\\
            0 & \text{otherwise}
        \end{cases}
    \end{equation}
    This forces the most indistinguishable configuration contributing to the distinguishability spectrum of $\rho\!\!\downharpoonright_{\partition{\Lambda}}$ to correspond to the partition state $\partket{\Lambda}$, thereby isolating a specific subset of errors impacting the original state.
    With access to this information, one can apply probabilistic error cancellation \cite{cai_quantum_2023} to mitigate distinguishability noise in the output distribution produced by the original state.
    Using $\rho_0$ to denote the ideal indistinguishable $n$-photon state, we have the following result:

    \begin{theorem}[Partition mitigation]
    \label{th:mitigation}
    Let $\rho\in \mathcal{I}$ with indistinguishabilities $M_\sigma$ and let $\{\rho\!\!\downharpoonright_{\partition{\Lambda}}\}$ be the set of its partitioned states.
    If $M_\sigma \neq 0$ for all $\sigma$, there exists a set of correction weights $w_{\partition{\Lambda}}$ such that:
    $$
    \sum_{\partition{\Lambda}} w_{\partition{\Lambda}} \rho\!\!\downharpoonright_{\partition{\Lambda}} \sim \rho_0.
    $$
        
    \end{theorem}
    \begin{proof}
        We can create a reconstruction matrix $R_{\partition{\Lambda}\partition{\Xi}}$ such that the column relative to the full partition $\{1,\hdots,n\}$ contains the indistinguishabilities $M_\sigma$ of $\rho$ and  the column relative to $\partition{\Xi}$ contains the indistinguishabilities of $\rho\!\!\downharpoonright_{\partition{\Xi}}$.
        The problem of finding the correction weights $w_{\partition{\Lambda}} $ that produce $M'_\sigma = 1$ for all $\sigma$ can be cast to the inversion of the following linear relations:
        $$
        1 = \sum_{\partition{\Xi}}R_{\partition{\Lambda}\partition{\Xi}} \cdot w_{\partition{\Xi}}  \quad \forall\: \partition{\Lambda}.
        $$
        Relying on the condition of \refeqn{eq:partitioning_condition} and sorting the partitions by number of subsets (see Lm.~\ref{lm:deferrec_celln}) the matrix $R_{\partition{\Lambda}\partition{\Xi}}$ can be arranged into a lower triangular form with the values $M_\sigma$ on the diagonal.
        By hypothesis no $M_\sigma$ is 0, $R_{\partition{\Lambda}\partition{\Xi}}$ is thus not singular and it can be inverted through forward substitution.
    \end{proof}

    {\begin{figure}[t]
        \includegraphics[width = .8\linewidth]{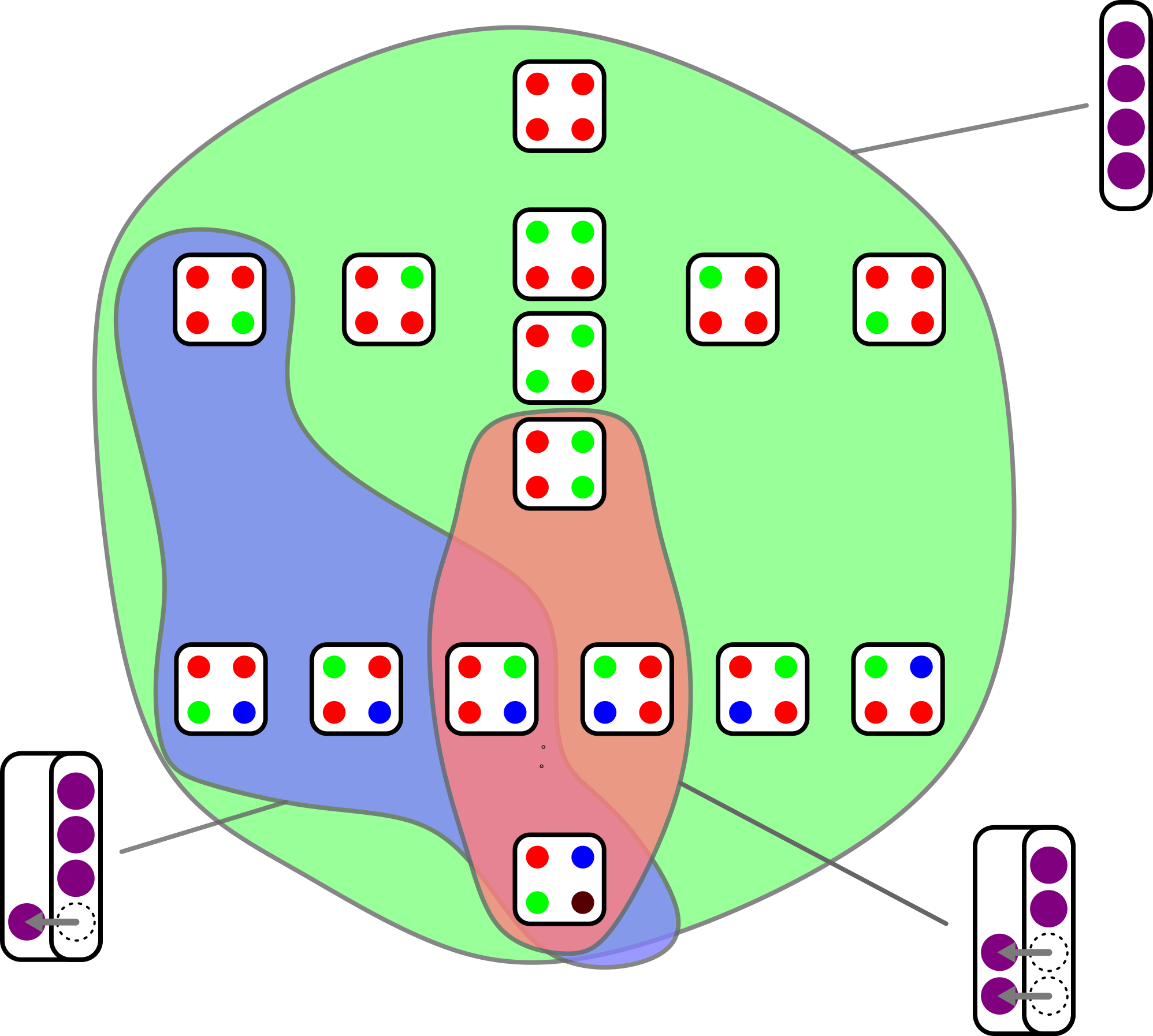}
        \caption{\textbf{Illustration of an error mitigation strategy for distinguishability errors with time-delay partitioning.} Applying discrete time delays to the input photons modifies the distinguishability spectrum produced by the original $n$-photon state (green area), providing access to subsets of errors arising from partial distinguishability.
        The incoherent distinguishability contributions expected from the partitioning $\{1,2,3\}\{4\}$ (blue area) and $\{1,2\}\{3,4\}$ (red area) are shown.}
        \label{fig:mitigation}
    \end{figure}
    
    Orbit invariance is a crucial requirement for this method to work: the number of distinct correction-experiments that can be realized by forcing complete distinguishability in some particles is limited by the number of distinct partitions. In order to be sure to achieve perfect correction, the state must have at most that number of free parameters.

    An interesting feature of this error mitigation procedure is that it can be implemented progressively with increasing orders of correction.
    The lower triangular matrix $R_{\partition{\Lambda}\partition{\Xi}}$ that guides the correction can be inverted one row at a time.
    Thus, restricting the procedure to a polynomial number of rows ensures that it depends only on a polynomial number of partitioning operations $\rho\!\!\downharpoonright_{\partition{\Lambda}}$, and therefore requires only a polynomial number of additional experiments and postprocessing steps—at the cost of achieving only partial correction.

    Along these same lines, we anticipate that the hierarchical structure of errors in the partition representation will prove valuable beyond error mitigation. In particular, it provides a scalable, stochastic framework for modeling partial distinguishability in contexts such as quantum error correction, where efficiently capturing low-weight errors is essential and the number of relevant contributions can remain polynomial in system size.

    \subsection{Incoherent distinguishability is reproducible}

    \label{ssec:twirling}
    Given the strict symmetry requirement imposed by orbit invariance it is quite easy to identify states that do not possess a partition representation (see for example Ex.~\ref{ex:triad_phase} in Appendix~\ref{app:examples}).
    Nevertheless, we shall show that it is always possible to recover a partition representation from any state.
    Given an arbitrary $n$-photon state, we can consider the probabilistic channel $\mathcal{T}[\cdot]$, constructed by averaging $\rho$ over the random applications of one of the $n!$ mode permutations: 
    \begin{definition}[Permutation Twirling]
    \label{def:proj_twirl}
        $$\mathcal{T}[\rho] = \frac{1}{n!}\sum_{\egperm{\sigma}} \opperm{\sigma} \rho \opperm{\sigma}^\dagger.$$
    \end{definition}
   
    By computing the generalized indistinguishability $M'_\sigma$ for the twirled state $\rho_\mathcal{T}$, we can see that this channel is able to enforce orbit invariance by averaging all indistinguishabilities with the same cycle structure:

    \begin{theorem}
    \label{th:twirling}
        Let $\rho\in \mathcal{D}$ with indistinguishabilities $M_\sigma$ and let $\rho_\mathcal{T} = \mathcal{T}[\rho]$ be its twirled state.
        Then $\rho_{\mathcal{T}} \in \mathcal{I}$.
    \end{theorem}

    \begin{proof}
        The indistinguishabilities $M'_\sigma$ of $\rho_\mathcal{T}$ read:
        
        \begin{equation}
            \label{eq:twirling}
            M'_\sigma = \frac{1}{n!}\sum_{\egperm{\tau}} \trace{\opperm{\tau^{-1}}\opperm{\sigma}\opperm{\tau}\rho} = \frac{1}{n!}\sum_{\egperm{\tau}} M_{\sigma^\tau},
        \end{equation}
        \noindent where the symbol $\sigma^\nu$ represents the conjugation $\nu \sigma \nu^{-1}$.

        The image of $\sigma^\nu$ for $\egperm{\nu}$ spans all the permutations with the same cycle structure as $\sigma$.
        Given two permutations $\partition{\sigma}$ and $ \partition{\sigma'}$ such that they can be ordered so that $|\partition{\sigma}_i| = |\partition{\sigma'}_i| $, it is always possible to find a third permutation $\nu$ that connects them through conjugation (see Lem.~\ref{lm:connect_conj})). Noting that, for fixed $\nu$, the two sums $\sum_{\nu\tau}$ and $\sum_{\tau}$ would be equivalent in \refeqn{eq:twirling}, we find that  $M'_\sigma = M'_{\sigma^\nu}$, hence orbit invariance is satisfied and consequently $\rho_\mathcal{T}$ has a partition representation.
    \end{proof}
    
    Reminiscent of `twirling' a qubit state with Pauli operators to render noise stochastic \cite{wallman_noise_2016}, we refer to this operation as permutation twirling.
    The averaging in \refeqn{eq:twirling} does more than just enforce orbit invariance: it effectively creates a situation similar to having  $\rho = \bigotimes_i^n \rho_0$ for some fixed single photon state $\rho_0$, where the generalized indistinguishabilities $M_\sigma$ only depend on the size of the cycles composing $\sigma$.
    
    The least destructive projection that enforces orbit invariance involves averaging only over permutations that share the same cycle partition as $\sigma$, i.e. those satisfying $\partition{\tau} = \partition{\sigma}$:
    
    \begin{definition} [Strict partition projection]
    \label{def:proj_strict}
        $$M'_\sigma = \frac{1}{\prod_i{(|\underline{\sigma_i}|-1)!}}\sum_{\partition{\tau} = \partition{\sigma}} M_\tau.$$.
    \end{definition}
    However, this strict partition projection is impossible to implement with random permutations, due to the fact that distinct partitions have incompatible stabilizers.
    If even physical, we believe it is challenging to devise an apparatus that realizes such an operation and we will leave this question for future work. 
    
    The permutation twirling operation can, on the other hand, be practically implemented via optical switches or more simply by averaging experimental runs while permuting the target unitary implemented by a reconfigurable photonic chip \cite{taballione_20mode_2023}, making it an appealing way to tailor distinguishability noise to the incoherent regime in Boson Sampling implementations.
    
 \subsection{Incoherent distinguishability retains hardness}
\label{sec:simulation}
    
    Steering the distinguishability spectrum of an $n$-photon state into the incoherent regime---whether via permutation twirling or a more strict projection---inevitably alters its interference behavior. It is therefore crucial to verify that such transformations do not compromise the state’s utility for quantum information processing.

    A first observation is that the set of states admitting a partition representation has a nontrivial intersection with states that are hard to simulate: the ideal indistinguishable state admits a simple partition representation, as do all states close to it in a partition-distribution sense.
    Therefore, under the hardness conjectures underlying Boson Sampling, there cannot exists a sampling algorithm that efficiently simulates all possible states in $\mathcal{I}$, regardless of the details of $\wpart{\Lambda}$. This rules out the possibility that this regime is fundamentally ``not hard".

    A projection onto this regime might nevertheless degrade its computational hardness by pushing it further from the ideal indistinguishable scenario than the original state.
    In order to exclude this scenario, we can pick a random no-collision outcome $\outcome$ and check how the ideal indistinguishable probability $p(s|\rho_0)$ compares to what is produced before $p(s|\rho)$ and after $p(s|\rho_{\mathcal{T}})$ the operation.
    Focusing on the permutation twirling channel $\mathcal{T}$, we have the following result:

    \begin{theorem}(Averaging effects)
        Let $\rho\in\mathcal{D}$, $\rho_{\mathcal{T}}\in \mathcal{I}$ be its twirled state, and $\rho_0$ an ideal indistinguishable state.
        Fixing a no-collision outcome $\outcome$ and averaging over the ensemble of Haar-random unitary transformations, the following inequality holds:
        $$\expected{(p(\outcome|\rho) - p(\outcome|\rho_0))^2} \geq \expected{(p(\outcome| \rho_{\mathcal{T}}) - p(\outcome|\rho_0))^2}.$$
    \end{theorem}

    A proof of this statement can be found in Appendix~\ref{app:l2dist}, and relies on the random matrix theory arguments detailed in Refs.~\cite{hoven_efficient_2024, renema_efficient_2018}.
    More generally, it is possible to show that most averaging operations over the indistinguishabilities of an $n$-photon state can actually be beneficial as the variance of $p(\outcome|\rho_0) - p(\outcome|\rho)$ can be shown to be an increasing function of the variability of $M_\sigma$ between $\sigma$ with the same cycle structure (specifically the cumulant $\var{M_\sigma}$).
    From this, it is clear that broader averaging operations are more rewarding compared to, for example, a strict partition projection with minimal averaging. Stricter projections would only attain an intermediate decrease in $\var{M_\sigma}$ compared to the twirled state.

\section{Boson Sampling with partition states}

    Partially distinguishable multi-photon states that are found in $\mathcal{I}$ exhibit, by definition, classical correlations. The observed distinguishability properties can be interpreted as the result of a random process $\markov$ which samples the distinguishability configurations $\partition{\Lambda}$ with weights $\wpart{\Lambda}$.
    Assuming that the sampling process $\markov$ is not inherently intractable, and the partition distribution is positive, it is possible to construct an exact sampling algorithm whose performance can be fully analyzed. Quasi-probability sampling techniques \cite{cai_quantum_2023} could also be used for distributions containing negative partition weights. It is interesting to note how even in this incoherent setting there still seems to be in general some added cost required to simulate the behavior of partial distinguishability, which could require knowledge of the full partition distribution or be affected by negativities.

    The core idea of the sampling algorithm is to exploit the fact that any partition state can be decomposed into indistinguishable parts that behave as smaller instances of the ideal sampling task (see Appendix~\ref{app:partstate}). After solving these smaller instances using a known algorithm, the results can be convoluted classically in a second step to reconstruct the full solution.
    On average, the correct statistics is obtained by sampling sequentially from each group of indistinguishable photons and then adding together each partial results.

    In Algorithm~\ref{alg:partition_sample}, we report a pseudo-code implementation of this partition sampling algorithm which assumes a subroutine $\markov\_sample()$ that draws a partition state for a given partition distribution $\{\wpart{\Lambda]}\}$ and $\Lambda\_sample()$ that samples an outcome vector from the output statistics of a given indistinguishable group of photons.

        \begin{algorithm}[H]
        \caption{Partition sampling algorithm}\label{alg:partition_sample}

        \begin{algorithmic}
        \Require{$\markov\_sample()$, $\Lambda\_sample()$}
            \Function{partition\_sample}{$ $}
    
                 \State $\underline{s} \gets [\underline{0}] $\Comment{initialize an empty sample}
                \State $\partition{\Lambda} \gets \markov\_sample()$ \Comment{Extract a partition}
                \For{ $\Lambda_i \in \partition{\Lambda}$}
                    \State $\underline{s_i} \gets \Lambda\_sample(\Lambda_i)$ \Comment{Obtain sample from subset} 
                    \State $\underline{s} \pluseq \underline{s_i}$ \Comment{Combine partial result} 
                \EndFor
                \State $\textbf{return } \underline{s}$
            \EndFunction
        \end{algorithmic}
        \end{algorithm}
        
    Assuming that $\Lambda\_sample()$ is carried out by the state-of-art algorithm described by Clifford and Clifford \cite{clifford_classical_2018,clifford_faster_2024}, the number of operations required to sample from a group of $k$ indistinguishable photons is $\mathcal{O}\!\left( {k2^k + \mathrm{poly}(m,k)} \right)$.
    Neglecting the polynomial overhead, which is irrelevant for the global scaling, we can bound the number of operations required to sample from any given partition configuration $\partition{\Lambda}$:

    \begin{equation}
        \Lambda\_sample(\partition{\Lambda}) \sim \sum_i |\Lambda_i| 2^{|\Lambda_i|} \leq n 2^{\norm{\Lambda}},
    \end{equation}

    \noindent where $\norm{\Lambda}$ is shorthand for $\mathrm{max}_i(|\Lambda_i|)$.
    This entails that the average-case performance will be linked to an exponential average of the biggest group of indistinguishable photons in the partition distribution.

    This average can be compute exactly for the common Orthogonal Bad Bits model \cite{sparrow_quantum_2017}, as shown here below:
    \begin{proposition}
        The average complexity of sampling from an OBB state with indistinguishability $x$ using the partition sampling algorithm is:
        $$ \mathcal{O}\!\left( n(1+x)^n \right).$$
    \end{proposition}
    \begin{proof}

        We recall that the partition distribution in the OBB model is the result of a coin-flip process that chooses each photon to be in the indistinguishable group with probability $x$ and completely distinguishable otherwise.
        From the above considerations, sorting the partition by the number $k$ of indistinguishable photons obtained, we have $\wpart{\Lambda} = x^k(1-x)^{n-k}$ with multiplicity $\binom{n}{k}$.
        For the OBB model, the group of indistinguishable photons is always the largest cell of the partition, so the relative bound for the number of operations is $n2^k$.
        Summing, one obtains the result:

        $$ \sum_{k=1}^n \binom{n}{k} x^k(1-x)^{n-k} \cdot n2^k = n (1+x)^n.$$
    \end{proof}
    
    From the point of view of partition sampling, the indistinguishability $x$ in the OBB model continuously tunes the sampling complexity from a permanent-like exponential scaling to a constant one in the classical regime. Surprisingly, no exact exponential-to-polynomial transition point is found before $x=0$.
%%%%%%%%%%%%%%%%%%%%%%%%%%%%%%%%%%%%%%%%%%%%%%%%%%%%%%%%%%%%%%%%%%%%%%%%%%%%%%%%
\section{Genuine indistinguishability}

\label{sec:disambiguation}

    Multiple preceding works \cite{brod_witnessing_2019,giordani_experimental_2020,pont_quantifying_2022} set the goal of quantifying ``genuine $n$-photon indistinguishability" (GI). This figure of merit is the share of an $n$-photon state that is attributable to $n$ fully indistinguishable photons. It quantifies, in the most intuitive way, the share of results from a photonic protocol that cannot be impacted by errors due to partial distinguishability.
    Measuring (or even bounding) the value of GI would therefore be a useful way to physically certify the correctness of a photonic computation.
    The concepts we have introduced in the earlier sections can now be used to address some outstanding issues that arise when defining this quantity.
    
    As originally stated in \cite{brod_witnessing_2019}, the concept of measuring GI relies on a decomposition of the state such as:

    \begin{equation}
        \label{eq:GI_decompose}
        \rho = \alpha_n \rho_n + \rho_{\mathrm{res}},
    \end{equation}

    \noindent where $\rho_n$ is a state which produces the statistics of $n$ indistinguishable photons and its weight $\alpha_n$ is what is referred to as GI.
    From the state-agnostic point of view of the distinguishabilities $M_\sigma$, this decomposition is ambiguous: there is no self-evident criterion that dictates how much of $\rho_n$ should be subtracted from $\rho$, or which properties should be satisfied by the distinguishabilities of the residual part $\rho_{\mathrm{res}}$.
    Physicality of $\rho_{\mathrm{res}}$ could be chosen as a basic criterion but in this case one faces a second issue, arising from Shchesnovich's framework: we do not yet posses a complete characterization of all set of $n!$ complex quantities, which corresponds to the generalized indistinguishabilities of a physical state.

    When a partition distribution exists, the partition representation naturally realizes the decomposition of \refeqn{eq:GI_decompose} and GI can be identified as the weight of the fully indistinguishable partition $p_{\{1,2,\dots,n\}}$.
    This broader definition of GI can be further extended to all possible $n$-photon states $\rho$: if $\rho$ does not have a partition representation, it can still be assigned the value of GI that is obtained after a projecting it onto the regime of incoherent distinguishability.
    In the context of the two projection operation we proposed, strict projection and permutation twirling, there is no ambiguity that arises from the choice of the projection as they both produce an average over all the generalized indistinguishabilities relative to maximal cycles:
    \begin{definition}
    \label{def:GI_part}
        $$\mathrm{GI}_{\mathrm{part}} = \frac{1}{(n-1)!}\sum_{\partition{\sigma} = \{1,2,\dots,n\}} M_\sigma.$$
    \end{definition}

    The subscript `part' is here used to highlight that this form of $\mathrm{GI}$ stems partition representation arguments.

    The effect of this choice of $\mathrm{GI}$ is to isolate all $n$-photon interference behavior of $\rho$ (as intended in Ref.~\cite{shchesnovich_collective_2018}) within $\rho_n$, at least on average, such that:
    
    \begin{equation}
        \sum_{\partition{\sigma} = \{1,2,\dots,n\}} \trace{\opperm{\sigma} \rho_\mathrm{res}} = 0.
    \end{equation}
    The hypothesis of orbit invariance here reflects the requirement that all such generalized indistinguishabilities have to be equal in order to be grouped together in $\rho_n$.

    From the literature, one can extract a second approach to GI that is inequivalent to the above definition.
    This approach, derived from representation theory arguments and found under various perspectives in Refs.~\cite{yurke_su2_1986,tan_su3_2013,tillmann_generalized_2015,khalid_permutational_2018,spivak_generalized_2022}, establishes the symmetric subspace projection as the measure of bosonic behavior.
    Given the subspace projector $\Pi_{\mathrm{sym}} = \frac{1}{n!}\sum_\sigma \opperm{\sigma}$, one can compute the modulus of this projection in terms of the generalized indistinguishabilities, leading us to consider a second definition of GI:
    \begin{definition}
    \label{def:GI_sym}
        $$\mathrm{GI}_{\mathrm{sym}} = \trace{\Pi_{\mathrm{sym}} \rho} = \frac{1}{n!} \sum_\sigma M_\sigma.$$
    \end{definition}

    Interestingly, this quantity has been found to exactly weight the share of probability relative to completely bunched outcomes \cite{tichy_sampling_2015,shchesnovich_partial_2015}. For this reason, one can argue that $\mathrm{GI}_{\mathrm{sym}}$ is actually the maximum value that GI can take without sacrificing physicality. Beyond this point, $\rho_\mathrm{res}$ would produce negative photocounting probabilities.
    
    When analyzing how this quantity has been defined in earlier works, we find a first group of works \cite{tichy_zerotransmission_2010,rodari_semideviceindependent_2025, deguise_coincidence_2014} that tried to access $\mathrm{GI}_\mathrm{sym}$, while a second group focused on the single indistinguishability $M_{(1\hdots n)}$ as a particular \cite{shchesnovich_collective_2018} or general \cite{karczewski_genuine_2019} quantifier of $n$-body behavior, or even making an explicit incoherent noise assumption \cite{brod_witnessing_2019,pont_quantifying_2022} and thus were aiming at the notion of $\mathrm{GI}_\mathrm{part}$. 

    Both definitions of GI reduce to different averages over the generalized indistinguishabilities and, whenever a state has a positive partition distribution, it is possible to directly compare the two quantities:
    
    \begin{theorem}
        Let $\rho\in \mathcal{I}$ with a positive partition representation $\wpart{\Lambda}\geq0$ and indistinguishabilities $M_\sigma$, then:
        $$\mathrm{GI}_{\mathrm{part}} = M_{(1...n)} \leq  \frac{1}{n!} \sum_{\egperm{\sigma}} M_\sigma =\mathrm{GI}_{\mathrm{sym}}.$$
    \end{theorem}
    \begin{proof}
        
    This bound follows from observing that positive partition weights imply that the generalized indistinguishabilities follow the hierarchy described by the order relation $\succeq$. More precisely, considering \refeqn{eq:msigma_sum1}, for two permutations $\sigma$ and $\tau$ with $\partition{\sigma} \succeq \partition{\tau}$, the transitive property ensures that the summation for $\sigma$ is included in the one over $\tau$. From the positivity of $\wpart{}$ we have that $ M_\sigma \leq M_\tau$.
    \end{proof}

    This inequality is strict whenever $ M_{(1...n)} \neq 1$ or, equivalently, unless $\rho$ is a perfectly indistinguishable $n$-photon state. Hence, they disagree whenever the state is nontrivial.
    
    Moreover, the two quantities can be computed and compared for a state following the Orthogonal Bad Bits model with average pairwise distinguishability $x^2$ (see Ex.~\ref{ex:obb} and \cite{rodari_semideviceindependent_2025}) obtaining:
    \begin{example}[OBB $\mathrm{GI}$]
    \label{ex:obb_GI}
    $$
        \mathrm{GI}_\mathrm{part} = x^n,\\
    $$
    $$
        \mathrm{GI}_\mathrm{sym} \approx x^ne^{\frac{1}{x}-1}.
    $$
    \end{example}

    Here the two definitions approximately agree in the indistinguishable regime where $x\approx 1$, yet they diverge exponentially in the distinguishable regime where $x\approx0$.
    This effect is even more evident if we analyze $\mathrm{GI}$ for a partition state with a single distinguishable particle $\partition{\Lambda}~=~\{1\}\{2,3,..n\}$:

    \begin{example}[Singly distinguishable $\mathrm{GI}$]
    \label{ex:single_dist}
    $$
        \mathrm{GI}_\mathrm{part} = 0,\\
    $$
    $$
        \mathrm{GI}_\mathrm{sym} = \frac{(n-1)!}{n!} = \frac{1}{n}.
    $$
    \end{example}

    It is thus clear that the two definitions generally quantify very different properties. The ambiguity highlighted in Ex.~\ref{ex:single_dist} reflects a tension between two perspectives: on one hand, the singly distinguishable state cannot support any task requiring genuine $n$-photon interference; on the other, the state plausibly retains some of the computational hardness of the fully indistinguishable state, preserving partial utility for sampling-based protocols.

    This incongruence stems from the absence of a proper resource theory for multi-photon interference, and we believe the task of defining GI could offer interesting clues to address this gap.
    %Along these lines, we have established a connection between the $\mathrm{GI}_\mathrm{part}$ and resources that can be extracted through incoherent operations or postprocessing (Thm.~\ref{th:mitigation}), whereas $\mathrm{GI}_\mathrm{sym}$ may be viewed as capturing a genuinely coherent resource---one that requires access to the symmetric subspace projection ($\Pi_{\mathrm{sym}}$).
    In order to produce a state that explicitly presents the relative indistinguishability resource, which remains ``latent" for an arbitrary state, the two quantifiers we have discussed require access to two distinct kinds of operations.
    As shown in Thm.~\ref{th:mitigation} reproducing $\mathrm{GI}_\mathrm{part}$ only requires classical evolutions, specifically random permutations, while reproducing $\mathrm{GI}_\mathrm{sym}$, would require implementing the coherent evolution $\Pi_{\mathrm{sym}}$.

%%%%%%%%%%%%%%%%%%%%%%%%%%%%%%%%%%%%%%%%%%%%%%%%%%%%%%%%%%%%%%%%%%%%%%%%%%%%%%%%
\section{Conclusions}
\label{sec:discussion}

% bridging the fecund concept of incoehrent errors to photon-native settings
% error mitigation strategy for distinguishability
% exactly as in the qubit picture there is a incoherent regime with an intuitive interpretation which can retain
% this work brings awareness on physical quantifiers of bosonic quality and calls for new measurement procedure and a more careful analysis of what's a rerousce for quantum 

% basis independent tool
In this article, we showcased how a general analysis of multi-photon interference can be carried out without relying on a density matrix representation.
Through Thm.~\ref{th:state_equivM}, we demonstrated that the interference behavior of an $n$-photon state can be fully characterized by measuring and tracking the evolution of the average values of the mode permutation operator $\avg{\opperm{\sigma}}$. In doing so, we constructed an explicit set of experiments that can be implemented to perform an effective tomography of distinguishability noise. This approach notably enabled a unified treatment of several works concerning multi-photon indistinguishability \cite{shchesnovich_collective_2018, spivak_generalized_2022, karczewski_genuine_2019, brod_witnessing_2019}, along with and their associated measurement protocols \cite{giordani_experimental_2018, pont_quantifying_2022, rodari_semideviceindependent_2025}.
While our analysis then focused on the subspace of incoherent errors, we believe this general description is of independent interest, as it may enable the exploration of other relevant scenario such as distinguishability distillation \cite{marshall_distillation_2022, somhorst_photon_2024}.

% general, basis independend condition for incoherent errors and analysis of the regime
We established a rigorous framework for the regime in which the imperfect interference behavior of an $n$-photon state can be fully described by incoherent distinguishability errors---that is, discrete, stochastic deviations from ideal indistinguishability (Thm.\ref{th:partition}). In this regime, the partition distribution $\wpart{\Lambda}$ serves as a complete and basis-independent description of interference, replacing the need to track all generalized indistinguishabilities $M_\sigma = \avg{\opperm{\sigma}}$ (Thm.\ref{th:equiv_part}).

% Noise tailoring has been proven a successful strategy in other areas of Quantum Computing\cite{wallman_noise_2016}: for what concerns modeling qubit errors, incoherent Pauli operations are preferable for their simple and predictable structure. We thus argue that there is a practical interest in pursuing an incoherent distinguishability regime for implementations of photonic computation.

Notably, we have shown that this incoherent regime enables an error mitigation strategy for distinguishability based on partition operations (Thm.~\ref{th:mitigation}), such as the controllable delay of groups of photons.
It is furthermore always possible to project an arbitrary state onto the incoherent regime through the application of random mode permutations (Thm.~\ref{th:twirling}), making error mitigation feasible even if the original state contains significant coherent errors.
Through arguments from random matrix theory used in Refs.~\cite{renema_efficient_2018,hoven_efficient_2024} to prove the simulability of some classes of partially distinguishable states, we evidenced that states projected onto the incoherent distinguishability regime retain hardness and are thus still useful for quantum information processing.
Moreover, the results of Refs.~\cite{renema_efficient_2018,hoven_efficient_2024} were recovered for a specific incoherent noise model and the concepts we discussed could be further used to broaden their scope.

Lastly, we discussed a problem intimately tied to incoherent distinguishability: the definition of ``genuine $n$-photon indistinguishability."
States that admit a partition representation naturally lead to a consistent definition of genuine $n$-photon indistinguishability in terms of a probability (Def.~\ref{def:GI_part}), which one can compare to an alternative definition that descends from symmetric subspace projection (Def.~\ref{def:GI_sym}).
Selected examples (Ex.\ref{ex:obb_GI},\ref{ex:single_dist}) highlighted regimes in which the two definitions diverge in their assessment of the resources provided by a state. While the partition framework helps expose this tension, our discussion underscores the need for further investigation to clarify which resource is actually degraded by distinguishability---and to what extent.

As a final note, with minor modifications to our arguments, an equivalent partition-based description of partially distinguishable fermions can be constructed. This complements the claims in Ref.~\cite{spivak_generalized_2022} where `Distinguishable Fermion Sampling' was proposed as a new, potentially hard, problem: whenever the effects of distinguishability can be simulated by a classical process, this cannot hold. In this context, a key question remains whether a `coherent' distinguishability regime makes the sampling problem harder---an open problem for both fermions and bosons.

\section*{Acknowledgements}

This work has been co-funded by the European Commission as part of the EIC accelerator program under the grant agreement 190188855 for SEPOQC project, by the Horizon-CL4 program under the grant agreement 101135288 for EPIQUE project, and by the CIFRE grant n°2024/0084.
E.~A. and S.~C.~W. would like to thank Shane Mansfield, Grégoire de Gliniasty, Rafail Frantzeskakis, and Samuel Mister for reading the manuscript and providing feedback; Rawad Mezher, James Mills, Hugo Thomas, Rodrigo Martinez, Danilo Triggiani, Marcin Karczewski, Olivier Krebs, Nadia Belabas, and Jean Senellart for helpful discussion and comments. 

\bibliography{bibtest}

%%%%%%%%%%%%%%%%%%%%%%%%%%%%%%%%%%%%%%%%%%%%%%%%%%%%%%%%%%%%%%%%%%%%%%%%%%%%%%%%
\appendix
\onecolumngrid

\section{Equivalence criterion from Shchesnovich's indistinguishabilities}
\label{app:cyclic}

    As originally derived in V. Shchesnovich's article~\cite{shchesnovich_partial_2015} and shown in Section~\ref{sec:preliminaries} of the main text (\refeqn{eq:path_interf}), the photocounting probabilities of an interference experiment involving an $n$-photon state can be expressed as a function of its $n!$ generalized indistinguishabilities $M_\sigma = \langle \opperm{\sigma}\rangle$.
    Grouping up the terms which depends on the same $M_\sigma$ we obtain:

    \begin{equation}
    \label{eq:path_interf_group}
    p(\outcome) = \sum_{\egperm{\sigma}} M_\sigma \sum_{\egperm{\tau}} X^*_{\tau}X_{\tau\sigma} .
    \end{equation}

    \noindent

    From this expression it follows that, if two states possess the same values of $M_\sigma$, they will always produce the same photocounting statistics in every interference experiment.
    
    \begin{lemma}[Shchesnovich's indistinguishabilities \cite{shchesnovich_partial_2015}]
    \label{lm:shch_necessary}
        Let $\rho,\rho' \in \mathcal{D}$ such that $M_\sigma=M'_\sigma$, then $p(\outcome| \rho) = p(\outcome|\rho')$ for every outcome $\outcome$ and linear interferometer $U$.
    \end{lemma}
    The backward implication can be proven by considering a family of $2n \times 2n$ interferometers which are able to isolate the dependence on one specific $M_\sigma$ in one specific outcome: if the photocounting statistics of two given states always coincide, one can proceed to measure each indistinguishability of the two states and the results must coincide.
    This family consists of $n!$ different interferometers, each one corresponding to a permutation $\sigma$ of the $n$ input photons and are notated by $C_\sigma$. These interferometers generalize the cyclic interferometers found in M. Pont et al.~\cite{pont_quantifying_2022}. In that work, they only considered the full cycle interferometer, notated here as $C_{(1~\!2~\!\dots~\!n)}$.
        
    The base structure consists of two layers of $n$ $50:50$ beam splitters separated by some pattern of mode swaps, and this structure is the same for all the interferometers (see Fig.\ref{fig:cyclic}). For clarity, we label the modes from 0 to $2n-1$, and the rows of beam splitters from 1 to $n$ so that, for example, row 1 contains modes 0 and 1 while row 2 contains modes 2 and 3.
    The even and odd modes between the first and second layer of beam splitters are linked as follows: the even output in row $i$ of first layer of beam splitters is connected directly to its corresponding even input mode in the same row of the second layer. The odd output mode on row $i$ on the first layer undergoes instead the permutation $\sigma$ specific to $C_\sigma$, being connected to the odd input mode in row $\sigma(i)$ in the second layer.
    For each disjoint cycle composing $\sigma_i$, a phase shifter imprinting a phase $\phi_i$ is positioned in the first even mode of the rows involved in that cycle.
    In Fig.~\ref{fig:cyclic}, we illustrate the interferometers $C_{(1243)}$ and $C_{(13)(24)}$.

        \begin{figure}[htb]
            
            \subfloat[\label{fig:dio}]{
            \begin{tikzpicture}[scale= 0.5,x=1pt,y=1pt]
                \draw[color=darkred,line width=3,line join=miter,fill=none] (10,-25) -- (25,-25);
                \draw[color=darkred,line width=3,line join=miter,fill=none] (10,-75) -- (25,-75);
                \draw[color=darkred,line width=3,line join=miter,fill=none] (10,-125) -- (25,-125);
                \draw[color=darkred,line width=3,line join=miter,fill=none] (10,-175) -- (25,-175);
                \draw[color=darkred,line width=3,line join=miter,fill=none] (10,-225) -- (25,-225);
                \draw[color=darkred,line width=3,line join=miter,fill=none] (10,-275) -- (25,-275);
                \draw[color=darkred,line width=3,line join=miter,fill=none] (10,-325) -- (25,-325);
                \draw[color=darkred,line width=3,line join=miter,fill=none] (10,-375) -- (25,-375);
                \draw[color=darkred,line width=3] (25,-25) -- (53,-25) -- (72,-44);
                \draw[color=darkred,line width=3] (78,-44) -- (97,-25) -- (125,-25);
                \draw[color=darkred,line width=3] (25,-75) -- (53,-75) -- (72,-56);
                \draw[color=darkred,line width=3] (78,-56) -- (97,-75) -- (125,-75);
                \draw[color=black,line width=1,fill=black] (50,-43) -- (100,-43) -- (100,-57) -- (50,-57) -- cycle;
                \node[align=center,anchor=base,font = {\fontsize{7pt}{0}\selectfont}] at (75,-85) {};
                \node[align=center,anchor=base,font = {\fontsize{7pt}{0}\selectfont}] at (75,-26) {};
                \draw[color=black,line width=1,fill=lightgray] (50,-43) -- (100,-43) -- (100,-47) -- (50,-47) -- cycle;
                \draw[color=darkred,line width=3] (25,-125) -- (53,-125) -- (72,-144);
                \draw[color=darkred,line width=3] (78,-144) -- (97,-125) -- (125,-125);
                \draw[color=darkred,line width=3] (25,-175) -- (53,-175) -- (72,-156);
                \draw[color=darkred,line width=3] (78,-156) -- (97,-175) -- (125,-175);
                \draw[color=black,line width=1,fill=black] (50,-143) -- (100,-143) -- (100,-157) -- (50,-157) -- cycle;
                \node[align=center,anchor=base,font = {\fontsize{7pt}{0}\selectfont}] at (75,-185) {};
                \node[align=center,anchor=base,font = {\fontsize{7pt}{0}\selectfont}] at (75,-126) {};
                \draw[color=black,line width=1,fill=lightgray] (50,-143) -- (100,-143) -- (100,-147) -- (50,-147) -- cycle;

                \draw[color=darkred,line width=3] (25,-225) -- (53,-225) -- (72,-244);
                \draw[color=darkred,line width=3] (78,-244) -- (97,-225) -- (125,-225);
                \draw[color=darkred,line width=3] (25,-275) -- (53,-275) -- (72,-256);
                \draw[color=darkred,line width=3] (78,-256) -- (97,-275) -- (125,-275);
                \draw[color=black,line width=1,fill=black] (50,-243) -- (100,-243) -- (100,-257) -- (50,-257) -- cycle;
                \node[align=center,anchor=base,font = {\fontsize{7pt}{0}\selectfont}] at (75,-285) {};
                \node[align=center,anchor=base,font = {\fontsize{7pt}{0}\selectfont}] at (75,-226) {};
                \draw[color=black,line width=1,fill=lightgray] (50,-243) -- (100,-243) -- (100,-247) -- (50,-247) -- cycle;

                \draw[color=darkred,line width=3] (25,-325) -- (53,-325) -- (72,-344);
                \draw[color=darkred,line width=3] (78,-344) -- (97,-325) -- (125,-325);
                \draw[color=darkred,line width=3] (25,-375) -- (53,-375) -- (72,-356);
                \draw[color=darkred,line width=3] (78,-356) -- (97,-375) -- (125,-375);
                \draw[color=black,line width=1,fill=black] (50,-343) -- (100,-343) -- (100,-357) -- (50,-357) -- cycle;
                \node[align=center,anchor=base,font = {\fontsize{7pt}{0}\selectfont}] at (75,-385) {};
                \node[align=center,anchor=base,font = {\fontsize{7pt}{0}\selectfont}] at (75,-326) {};
                \draw[color=black,line width=1,fill=lightgray] (50,-343) -- (100,-343) -- (100,-347) -- (50,-347) -- cycle;

                \draw[color=darkred,line width=3] (125,-25) -- (175,-25);
                \draw[color=black,line width=1,fill=gray] (130,-40) -- (139,-40) -- (153,-10) -- (144,-10) -- (130,-40) -- (139,-40) -- cycle;
                \node[align=center,anchor=west,font = {\fontsize{12pt}{0}\selectfont}] at (147,-38) {$\phi_1$};
                \draw[color=darkred,line width=3] (125,-75) -- (175,-75);
                \draw[color=darkred,line width=3] (125,-125) -- (175,-125);
                \draw[color=darkred,line width=3] (125,-175) -- (175,-175);
                \draw[color=darkred,line width=3] (125,-225) -- (175,-225);
                \draw[color=darkred,line width=3] (125,-275) -- (175,-275);
                \draw[color=darkred,line width=3] (125,-325) -- (175,-325);
                \draw[color=darkred,line width=3] (125,-375) -- (175,-375);
                \draw[color=white,line width=6] (178,-25) -- (222,-25);
                \draw[color=darkred,line width=3] (175,-25) -- (178,-25) -- (222,-25) -- (225,-25);
                \draw[color=white,line width=6] (178,-75) -- (222,-175);
                \draw[color=darkred,line width=3] (175,-75) -- (178,-75) -- (222,-175) -- (225,-175);
                \draw[color=white,line width=6] (178,-125) -- (222,-125);
                \draw[color=darkred,line width=3] (175,-125) -- (178,-125) -- (222,-125) -- (225,-125);
                \draw[color=white,line width=6] (178,-175) -- (222,-375);
                \draw[color=darkred,line width=3] (175,-175) -- (178,-175) -- (222,-375) -- (225,-375);
                \draw[color=white,line width=6] (178,-225) -- (222,-225);
                \draw[color=darkred,line width=3] (175,-225) -- (178,-225) -- (222,-225) -- (225,-225);
                \draw[color=white,line width=6] (178,-275) -- (222,-75);
                \draw[color=darkred,line width=3] (175,-275) -- (178,-275) -- (222,-75) -- (225,-75);
                \draw[color=white,line width=6] (178,-325) -- (222,-325);
                \draw[color=darkred,line width=3] (175,-325) -- (178,-325) -- (222,-325) -- (225,-325);
                \draw[color=white,line width=6] (178,-375) -- (222,-275);
                \draw[color=darkred,line width=3] (175,-375) -- (178,-375) -- (222,-275) -- (225,-275);
                \draw[color=darkred,line width=3] (225,-25) -- (253,-25) -- (272,-44);
                \draw[color=darkred,line width=3] (278,-44) -- (297,-25) -- (325,-25);
                \draw[color=darkred,line width=3] (225,-75) -- (253,-75) -- (272,-56);
                \draw[color=darkred,line width=3] (278,-56) -- (297,-75) -- (325,-75);
                \draw[color=black,line width=1,fill=black] (250,-43) -- (300,-43) -- (300,-57) -- (250,-57) -- cycle;
                \node[align=center,anchor=base,font = {\fontsize{7pt}{0}\selectfont}] at (275,-85) {};
                \node[align=center,anchor=base,font = {\fontsize{7pt}{0}\selectfont}] at (275,-26) {};
                \draw[color=black,line width=1,fill=lightgray] (250,-43) -- (300,-43) -- (300,-47) -- (250,-47) -- cycle;

                \draw[color=darkred,line width=3] (225,-125) -- (253,-125) -- (272,-144);
                \draw[color=darkred,line width=3] (278,-144) -- (297,-125) -- (325,-125);
                \draw[color=darkred,line width=3] (225,-175) -- (253,-175) -- (272,-156);
                \draw[color=darkred,line width=3] (278,-156) -- (297,-175) -- (325,-175);
                \draw[color=black,line width=1,fill=black] (250,-143) -- (300,-143) -- (300,-157) -- (250,-157) -- cycle;
                \node[align=center,anchor=base,font = {\fontsize{7pt}{0}\selectfont}] at (275,-185) {};
                \node[align=center,anchor=base,font = {\fontsize{7pt}{0}\selectfont}] at (275,-126) {};
                \draw[color=black,line width=1,fill=lightgray] (250,-143) -- (300,-143) -- (300,-147) -- (250,-147) -- cycle;
                \draw[color=darkred,line width=3] (225,-225) -- (253,-225) -- (272,-244);
                \draw[color=darkred,line width=3] (278,-244) -- (297,-225) -- (325,-225);
                \draw[color=darkred,line width=3] (225,-275) -- (253,-275) -- (272,-256);
                \draw[color=darkred,line width=3] (278,-256) -- (297,-275) -- (325,-275);
                \draw[color=black,line width=1,fill=black] (250,-243) -- (300,-243) -- (300,-257) -- (250,-257) -- cycle;
                \node[align=center,anchor=base,font = {\fontsize{7pt}{0}\selectfont}] at (275,-285) {};
                \node[align=center,anchor=base,font = {\fontsize{7pt}{0}\selectfont}] at (275,-226) {};
                \draw[color=black,line width=1,fill=lightgray] (250,-243) -- (300,-243) -- (300,-247) -- (250,-247) -- cycle;
                \draw[color=darkred,line width=3] (225,-325) -- (253,-325) -- (272,-344);
                \draw[color=darkred,line width=3] (278,-344) -- (297,-325) -- (325,-325);
                \draw[color=darkred,line width=3] (225,-375) -- (253,-375) -- (272,-356);
                \draw[color=darkred,line width=3] (278,-356) -- (297,-375) -- (325,-375);
                \draw[color=black,line width=1,fill=black] (250,-343) -- (300,-343) -- (300,-357) -- (250,-357) -- cycle;
                \node[align=center,anchor=base,font = {\fontsize{7pt}{0}\selectfont}] at (275,-385) {};
                \node[align=center,anchor=base,font = {\fontsize{7pt}{0}\selectfont}] at (275,-326) {};
                \draw[color=black,line width=1,fill=lightgray] (250,-343) -- (300,-343) -- (300,-347) -- (250,-347) -- cycle;
                \draw[color=darkred,line width=3,line join=miter,fill=none] (325,-25) -- (340,-25);
                \draw[color=darkred,line width=3,line join=miter,fill=none] (325,-75) -- (340,-75);
                \draw[color=darkred,line width=3,line join=miter,fill=none] (325,-125) -- (340,-125);
                \draw[color=darkred,line width=3,line join=miter,fill=none] (325,-175) -- (340,-175);
                \draw[color=darkred,line width=3,line join=miter,fill=none] (325,-225) -- (340,-225);
                \draw[color=darkred,line width=3,line join=miter,fill=none] (325,-275) -- (340,-275);
                \draw[color=darkred,line width=3,line join=miter,fill=none] (325,-325) -- (340,-325);
                \draw[color=darkred,line width=3,line join=miter,fill=none] (325,-375) -- (340,-375);
                \node[align=center,anchor=east,font = {\fontsize{12pt}{0}\selectfont}] at (370,-28) {0};
                \node[align=center,anchor=east,font = {\fontsize{12pt}{0}\selectfont}] at (370,-78) {1};
                \node[align=center,anchor=east,font = {\fontsize{12pt}{0}\selectfont}] at (370,-128) {2};
                \node[align=center,anchor=east,font = {\fontsize{12pt}{0}\selectfont}] at (370,-178) {3};
                \node[align=center,anchor=east,font = {\fontsize{12pt}{0}\selectfont}] at (370,-228) {4};
                \node[align=center,anchor=east,font = {\fontsize{12pt}{0}\selectfont}] at (370,-278) {5};
                \node[align=center,anchor=east,font = {\fontsize{12pt}{0}\selectfont}] at (370,-328) {6};
                \node[align=center,anchor=east,font = {\fontsize{12pt}{0}\selectfont}] at (370,-378) {7};
                \node[align=center,anchor=west,font = {\fontsize{12pt}{0}\selectfont}] at (-20,-28) {0};
                \node[align=center,anchor=west,font = {\fontsize{12pt}{0}\selectfont}] at (-20,-78) {1};
                \node[align=center,anchor=west,font = {\fontsize{12pt}{0}\selectfont}] at (-20,-128) {2};
                \node[align=center,anchor=west,font = {\fontsize{12pt}{0}\selectfont}] at (-20,-178) {3};
                \node[align=center,anchor=west,font = {\fontsize{12pt}{0}\selectfont}] at (-20,-228) {4};
                \node[align=center,anchor=west,font = {\fontsize{12pt}{0}\selectfont}] at (-20,-278) {5};
                \node[align=center,anchor=west,font = {\fontsize{12pt}{0}\selectfont}] at (-20,-328) {6};
                \node[align=center,anchor=west,font = {\fontsize{12pt}{0}\selectfont}] at (-20,-378) {7};
                \end{tikzpicture}
        }
        \hspace*{\fill}
        \subfloat[\label{fig:cane}]{
            \begin{tikzpicture}[scale= 0.5,x=1pt,y=1pt]
                \draw[color=darkred,line width=3,line join=miter,fill=none] (10,-25) -- (25,-25);
                \draw[color=darkred,line width=3,line join=miter,fill=none] (10,-75) -- (25,-75);
                \draw[color=darkred,line width=3,line join=miter,fill=none] (10,-125) -- (25,-125);
                \draw[color=darkred,line width=3,line join=miter,fill=none] (10,-175) -- (25,-175);
                \draw[color=darkred,line width=3,line join=miter,fill=none] (10,-225) -- (25,-225);
                \draw[color=darkred,line width=3,line join=miter,fill=none] (10,-275) -- (25,-275);
                \draw[color=darkred,line width=3,line join=miter,fill=none] (10,-325) -- (25,-325);
                \draw[color=darkred,line width=3,line join=miter,fill=none] (10,-375) -- (25,-375);
                \draw[color=darkred,line width=3] (25,-25) -- (53,-25) -- (72,-44);
                \draw[color=darkred,line width=3] (78,-44) -- (97,-25) -- (125,-25);
                \draw[color=darkred,line width=3] (25,-75) -- (53,-75) -- (72,-56);
                \draw[color=darkred,line width=3] (78,-56) -- (97,-75) -- (125,-75);
                \draw[color=black,line width=1,fill=black] (50,-43) -- (100,-43) -- (100,-57) -- (50,-57) -- cycle;
                \node[align=center,anchor=base,font = {\fontsize{7pt}{0}\selectfont}] at (75,-85) {};
                \node[align=center,anchor=base,font = {\fontsize{7pt}{0}\selectfont}] at (75,-26) {};
                \draw[color=black,line width=1,fill=lightgray] (50,-43) -- (100,-43) -- (100,-47) -- (50,-47) -- cycle;
                \draw[color=darkred,line width=3] (25,-125) -- (53,-125) -- (72,-144);
                \draw[color=darkred,line width=3] (78,-144) -- (97,-125) -- (125,-125);
                \draw[color=darkred,line width=3] (25,-175) -- (53,-175) -- (72,-156);
                \draw[color=darkred,line width=3] (78,-156) -- (97,-175) -- (125,-175);
                \draw[color=black,line width=1,fill=black] (50,-143) -- (100,-143) -- (100,-157) -- (50,-157) -- cycle;
                \node[align=center,anchor=base,font = {\fontsize{7pt}{0}\selectfont}] at (75,-185) {};
                \node[align=center,anchor=base,font = {\fontsize{7pt}{0}\selectfont}] at (75,-126) {};
                \draw[color=black,line width=1,fill=lightgray] (50,-143) -- (100,-143) -- (100,-147) -- (50,-147) -- cycle;
                \draw[color=darkred,line width=3] (25,-225) -- (53,-225) -- (72,-244);
                \draw[color=darkred,line width=3] (78,-244) -- (97,-225) -- (125,-225);
                \draw[color=darkred,line width=3] (25,-275) -- (53,-275) -- (72,-256);
                \draw[color=darkred,line width=3] (78,-256) -- (97,-275) -- (125,-275);
                \draw[color=black,line width=1,fill=black] (50,-243) -- (100,-243) -- (100,-257) -- (50,-257) -- cycle;
                \node[align=center,anchor=base,font = {\fontsize{7pt}{0}\selectfont}] at (75,-285) {};
                \node[align=center,anchor=base,font = {\fontsize{7pt}{0}\selectfont}] at (75,-226) {};
                \draw[color=black,line width=1,fill=lightgray] (50,-243) -- (100,-243) -- (100,-247) -- (50,-247) -- cycle;
                \draw[color=darkred,line width=3] (25,-325) -- (53,-325) -- (72,-344);
                \draw[color=darkred,line width=3] (78,-344) -- (97,-325) -- (125,-325);
                \draw[color=darkred,line width=3] (25,-375) -- (53,-375) -- (72,-356);
                \draw[color=darkred,line width=3] (78,-356) -- (97,-375) -- (125,-375);
                \draw[color=black,line width=1,fill=black] (50,-343) -- (100,-343) -- (100,-357) -- (50,-357) -- cycle;
                \node[align=center,anchor=base,font = {\fontsize{7pt}{0}\selectfont}] at (75,-385) {};
                \node[align=center,anchor=base,font = {\fontsize{7pt}{0}\selectfont}] at (75,-326) {};
                \draw[color=black,line width=1,fill=lightgray] (50,-343) -- (100,-343) -- (100,-347) -- (50,-347) -- cycle;
                \draw[color=darkred,line width=3] (125,-25) -- (175,-25);
                \draw[color=black,line width=1,fill=gray] (130,-40) -- (139,-40) -- (153,-10) -- (144,-10) -- (130,-40) -- (139,-40) -- cycle;
                \node[align=center,anchor=west,font = {\fontsize{12pt}{0}\selectfont}] at (147,-38) {$\phi_1$};
                \draw[color=darkred,line width=3] (125,-225) -- (175,-225);
                \draw[color=black,line width=1,fill=gray] (130,-240) -- (139,-240) -- (153,-210) -- (144,-210) -- (130,-240) -- (139,-240) -- cycle;
                \node[align=center,anchor=west,font = {\fontsize{12pt}{0}\selectfont}] at (147,-238) {$\phi_2$};
                \draw[color=darkred,line width=3] (125,-75) -- (175,-75);
                \draw[color=darkred,line width=3] (125,-125) -- (175,-125);
                \draw[color=darkred,line width=3] (125,-175) -- (175,-175);
                \draw[color=darkred,line width=3] (125,-275) -- (175,-275);
                \draw[color=darkred,line width=3] (125,-325) -- (175,-325);
                \draw[color=darkred,line width=3] (125,-375) -- (175,-375);
                \draw[color=white,line width=6] (178,-25) -- (222,-25);
                \draw[color=darkred,line width=3] (175,-25) -- (178,-25) -- (222,-25) -- (225,-25);
                \draw[color=white,line width=6] (178,-75) -- (222,-175);
                \draw[color=darkred,line width=3] (175,-75) -- (178,-75) -- (222,-175) -- (225,-175);
                \draw[color=white,line width=6] (178,-125) -- (222,-125);
                \draw[color=darkred,line width=3] (175,-125) -- (178,-125) -- (222,-125) -- (225,-125);
                \draw[color=white,line width=6] (178,-175) -- (222,-75);
                \draw[color=darkred,line width=3] (175,-175) -- (178,-175) -- (222,-75) -- (225,-75);
                \draw[color=white,line width=6] (178,-225) -- (222,-225);
                \draw[color=darkred,line width=3] (175,-225) -- (178,-225) -- (222,-225) -- (225,-225);
                \draw[color=white,line width=6] (178,-275) -- (222,-375);
                \draw[color=darkred,line width=3] (175,-275) -- (178,-275) -- (222,-375) -- (225,-375);
                \draw[color=white,line width=6] (178,-325) -- (222,-325);
                \draw[color=darkred,line width=3] (175,-325) -- (178,-325) -- (222,-325) -- (225,-325);
                \draw[color=white,line width=6] (178,-375) -- (222,-275);
                \draw[color=darkred,line width=3] (175,-375) -- (178,-375) -- (222,-275) -- (225,-275);
                \draw[color=darkred,line width=3] (225,-25) -- (253,-25) -- (272,-44);
                \draw[color=darkred,line width=3] (278,-44) -- (297,-25) -- (325,-25);
                \draw[color=darkred,line width=3] (225,-75) -- (253,-75) -- (272,-56);
                \draw[color=darkred,line width=3] (278,-56) -- (297,-75) -- (325,-75);
                \draw[color=black,line width=1,fill=black] (250,-43) -- (300,-43) -- (300,-57) -- (250,-57) -- cycle;
                \node[align=center,anchor=base,font = {\fontsize{7pt}{0}\selectfont}] at (275,-85) {};
                \node[align=center,anchor=base,font = {\fontsize{7pt}{0}\selectfont}] at (275,-26) {};
                \draw[color=black,line width=1,fill=lightgray] (250,-43) -- (300,-43) -- (300,-47) -- (250,-47) -- cycle;
                \draw[color=darkred,line width=3] (225,-125) -- (253,-125) -- (272,-144);
                \draw[color=darkred,line width=3] (278,-144) -- (297,-125) -- (325,-125);
                \draw[color=darkred,line width=3] (225,-175) -- (253,-175) -- (272,-156);
                \draw[color=darkred,line width=3] (278,-156) -- (297,-175) -- (325,-175);
                \draw[color=black,line width=1,fill=black] (250,-143) -- (300,-143) -- (300,-157) -- (250,-157) -- cycle;
                \node[align=center,anchor=base,font = {\fontsize{7pt}{0}\selectfont}] at (275,-185) {};
                \node[align=center,anchor=base,font = {\fontsize{7pt}{0}\selectfont}] at (275,-126) {};
                \draw[color=black,line width=1,fill=lightgray] (250,-143) -- (300,-143) -- (300,-147) -- (250,-147) -- cycle;
                \draw[color=darkred,line width=3] (225,-225) -- (253,-225) -- (272,-244);
                \draw[color=darkred,line width=3] (278,-244) -- (297,-225) -- (325,-225);
                \draw[color=darkred,line width=3] (225,-275) -- (253,-275) -- (272,-256);
                \draw[color=darkred,line width=3] (278,-256) -- (297,-275) -- (325,-275);
                \draw[color=black,line width=1,fill=black] (250,-243) -- (300,-243) -- (300,-257) -- (250,-257) -- cycle;
                \node[align=center,anchor=base,font = {\fontsize{7pt}{0}\selectfont}] at (275,-285) {};
                \node[align=center,anchor=base,font = {\fontsize{7pt}{0}\selectfont}] at (275,-226) {};
                \draw[color=black,line width=1,fill=lightgray] (250,-243) -- (300,-243) -- (300,-247) -- (250,-247) -- cycle;
                \draw[color=darkred,line width=3] (225,-325) -- (253,-325) -- (272,-344);
                \draw[color=darkred,line width=3] (278,-344) -- (297,-325) -- (325,-325);
                \draw[color=darkred,line width=3] (225,-375) -- (253,-375) -- (272,-356);
                \draw[color=darkred,line width=3] (278,-356) -- (297,-375) -- (325,-375);
                \draw[color=black,line width=1,fill=black] (250,-343) -- (300,-343) -- (300,-357) -- (250,-357) -- cycle;
                \node[align=center,anchor=base,font = {\fontsize{7pt}{0}\selectfont}] at (275,-385) {};
                \node[align=center,anchor=base,font = {\fontsize{7pt}{0}\selectfont}] at (275,-326) {};
                \draw[color=black,line width=1,fill=lightgray] (250,-343) -- (300,-343) -- (300,-347) -- (250,-347) -- cycle;
                \draw[color=darkred,line width=3,line join=miter,fill=none] (325,-25) -- (340,-25);
                \draw[color=darkred,line width=3,line join=miter,fill=none] (325,-75) -- (340,-75);
                \draw[color=darkred,line width=3,line join=miter,fill=none] (325,-125) -- (340,-125);
                \draw[color=darkred,line width=3,line join=miter,fill=none] (325,-175) -- (340,-175);
                \draw[color=darkred,line width=3,line join=miter,fill=none] (325,-225) -- (340,-225);
                \draw[color=darkred,line width=3,line join=miter,fill=none] (325,-275) -- (340,-275);
                \draw[color=darkred,line width=3,line join=miter,fill=none] (325,-325) -- (340,-325);
                \draw[color=darkred,line width=3,line join=miter,fill=none] (325,-375) -- (340,-375);
                \node[align=center,anchor=east,font = {\fontsize{12pt}{0}\selectfont}] at (370,-28) {0};
                \node[align=center,anchor=east,font = {\fontsize{12pt}{0}\selectfont}] at (370,-78) {1};
                \node[align=center,anchor=east,font = {\fontsize{12pt}{0}\selectfont}] at (370,-128) {2};
                \node[align=center,anchor=east,font = {\fontsize{12pt}{0}\selectfont}] at (370,-178) {3};
                \node[align=center,anchor=east,font = {\fontsize{12pt}{0}\selectfont}] at (370,-228) {4};
                \node[align=center,anchor=east,font = {\fontsize{12pt}{0}\selectfont}] at (370,-278) {5};
                \node[align=center,anchor=east,font = {\fontsize{12pt}{0}\selectfont}] at (370,-328) {6};
                \node[align=center,anchor=east,font = {\fontsize{12pt}{0}\selectfont}] at (370,-378) {7};
                \node[align=center,anchor=west,font = {\fontsize{12pt}{0}\selectfont}] at (-20,-28) {0};
                \node[align=center,anchor=west,font = {\fontsize{12pt}{0}\selectfont}] at (-20,-78) {1};
                \node[align=center,anchor=west,font = {\fontsize{12pt}{0}\selectfont}] at (-20,-128) {2};
                \node[align=center,anchor=west,font = {\fontsize{12pt}{0}\selectfont}] at (-20,-178) {3};
                \node[align=center,anchor=west,font = {\fontsize{12pt}{0}\selectfont}] at (-20,-228) {4};
                \node[align=center,anchor=west,font = {\fontsize{12pt}{0}\selectfont}] at (-20,-278) {5};
                \node[align=center,anchor=west,font = {\fontsize{12pt}{0}\selectfont}] at (-20,-328) {6};
                \node[align=center,anchor=west,font = {\fontsize{12pt}{0}\selectfont}] at (-20,-378) {7};
                \end{tikzpicture}
            }
            \caption{The interferometers $C_{(1243)}$ (a) and $C_{(12)(34)}$ (b). For both setups the input 4-photon state is fed in the mode 0,2,4,6 and the sought outcome is a coincidence in the output modes 0,2,4,6. The image has been generated using the Python package Perceval \cite{heurtel_perceval_2023} }
            \label{fig:cyclic}
        \end{figure}
        
    To measure $M_\sigma$, the $n$ photons are fed into the $n$ even input modes of the interferometer. For every configuration, we need to measure only the standard outcome $\outcome_0 = [1,0,1,0...1,0]$, which is the one where the photons have not changed mode. The interferometers (and their scattering matrices $U$) are sparse enough so that, in the simple case where $\sigma$ has a single cycle, such an outcome can only be realized by two virtual paths: $X_{\mathrm{id}}$ where all the photons end up in the top arms after the first layer of beam splitter (and thus do not switch mode) or $X_\sigma$ where they all end up in the bottom arms and perform the permutation.
    The interference between these two paths creates a $\phi_1$-dependent fringe in $\outcome_0$, isolating the two conjugate generalized indistinguishabilities $M_\sigma = M_{\sigma^{-1}}^*$ and resulting in the probability formula derived in Ref.~\cite{pont_quantifying_2022}.
    Highlighting modulus and phase of $M_\sigma = |M_\sigma|e^{i \theta_\sigma}$ one has:

    \begin{equation}
        \label{eq:cyclic1}
        p(\outcome_0) = \frac{1}{2^{2n-1}}(1 \pm |M_\sigma|\text{cos}(\theta_\sigma + \phi_1 )),
    \end{equation}
    \noindent where the sign of the sum is a function of the size of the interferometer.
    Fitting the amplitude and phase of the fringe function allows the direct observation of $M_\sigma$.
    
    For the general case where $\sigma = \sigma_1 \sigma_2 \dots \sigma_k$ has more than one disjoint cycle, the interferometer is broken into $k$ independent parts, each one allowing again only an identity and a permutation path: $ X_{\sigma_i^{0}}, X_{\sigma_i^{1}}$.
    There thus a total of $2^k$ $n$-photon paths, labeled by a vector of $k$ choices $\underline{c}\in \{0,1\}^k$, whose probability amplitudes can be recovered by multiplying the amplitude of each partial path: 
    \begin{equation}
        X_{\underline{c}}  = \prod_{i=1}^k X_{\sigma_i^{c_i}}.
    \end{equation}

    Given this structure, one can derive that the experiment is now susceptible $3^{k-r}2^r$ distinct indistinguishabilities, where $r$ is the number of cycles $\sigma_i$ of order 2. The permutations $\tau_{\underline{l}}$ relative to these indistinguishabilities take the form:

    \begin{equation}
        \tau_{\underline{l}} = \prod_{i=1}^k \sigma_i^{l_i},
    \end{equation}

    \noindent where $l_i \in \{0,1\}$ if  $\sigma_i$ has order $2$ and $l_i \in \{-1,0,1\}$ otherwise.

    Noting that $|X_{\underline{c}}| = \frac{1}{2^n}$, the probability formula for $\outcome_0$ can be expressed as:
    
    \begin{equation}
        \label{eq:cyclick}
        p(\outcome_0) = \sum_{\underline{c},\underline{c'} \in \{0,1\}^k}  \pm \frac{|M_{\tau_{\underline{c}-\underline{c}'}}|}{2^{2n}} \cos\left(\theta_{\tau_{\underline{c}-\underline{c}'}} + \sum_i (c_i - c'_i)\phi_i\right),
    \end{equation}

    \noindent where now the sign of each term sum is also a function of the length of each cycle.
    The two probability equations (\ref{eq:cyclic1}) (\ref{eq:cyclick}) are the elements needed for the proof of the equivalence criterion.

    \begin{lemma}
    \label{lm:shch_sufficient}
        Let $\rho,\rho' \in \mathcal{D}$ with generalized indistinguishabilities $M_\sigma$ and $M'_\sigma$.
        If $p(\outcome| \rho) = p(\outcome|\rho')$ for every outcome $\outcome$ and linear interferometer $U$, then $M_\sigma$ = $M'_\sigma$ for all $\sigma$.
    \end{lemma}
    \begin{proof}
        If sigma has just one cycle the value of $M_\sigma$ is directly observable with the interferometer $C_\sigma$ because of \refeqn{eq:cyclic1}.

        Let $\sigma$ be a permutation with $k$ cycles.
        One can construct a one-parameter fringe function in the interferometer $C_\sigma$ by setting all the internal phases of the interferometer to $\phi_i = \overline{\phi} \in [0,2\pi)$ so that the path relative to $\underline{c}=[0,0,\dots,0]$ gains a phase proportional to $k\overline{\phi}$.
        Consequently, out of all the terms in the summation of \refeqn{eq:cyclick}, the two terms relative to $\underline{c} = [0,0,\dots,0], [1,1,\dots,1]$, linked to the two conjugate indistinguishabilities $M_\sigma = M^*_{\sigma^{-1}}$, are the only one producing a fringe proportional to $\cos(\theta_\sigma + k\overline{\phi})$.
        The modulus and phase of $M_\sigma$ are then observable as the highest-frequency component of the fringe function after a Fourier transform.
        Repeating the operation for all $\sigma$ with $k$ cycles allows the direct observation of all such indistinguishabilities.
    \end{proof}
%%%%%%%%%%%%%%%%%%%%%%%%%%%%%%%%%%%%%%%%%%%%%%%%%%%%%%%%%%%%%%%%%%%%%%%%%%%%%%%%

\section{Partition, permutations and ordering}
\label{app:partition}
    Given a set of $n$ elements $\Omega = \{1,2,\dots, n\}$ we define a partition $\partition{\Lambda}$ as a collection of non-intersecting subsets $ \Lambda_i \subseteq \Omega$ (cells) such that:
    \begin{equation}
        \bigcup_i \Lambda_i = \Omega.
    \end{equation}

    The number of subsets is indicated by $|\partition{\Lambda}|$ and the size of each subset by $|\Lambda_i|$.
    The number of distinct partitions of a set is counted by the $n$-th Bell number $\belln{n}$ \cite{rota_number_1964}, symbol that with a slight abuse will also be used to address the set of partitions ($\partition{\Lambda} \in \belln{n}$).
    Partitions can be naturally ordered by set inclusion, this results in a partial order relation that we denote with $\succeq$ and we define as follows:
    
    \begin{definition}
        Given two partitions $\partition{\Lambda}$ and $\partition{\Xi}$, \\
        $\partition{\Lambda} \succeq \partition{\Xi}$ $\iff$ $ \forall$ $\Xi_i \in \partition{\Xi}$ $\exists$ $\Lambda_j \in \partition{\Lambda}$ such that $\Lambda_j \supseteq \Xi_i$.
    \end{definition}

    By overlapping the two relations $\succeq$ and $\preceq$, one can rule out the possibility that two distinct partitions can be labeled as equal:

    \begin{lemma}[Cell number]
    \label{lm:deferrec_celln}
        $\partition{\Lambda} \succ \partition{\Xi} \follows  |\partition{\Lambda}| < |\partition{\Xi}|$.
    \end{lemma}
    \begin{proof}
        The definition of $\partition{\Lambda} \succeq \partition{\Xi}$ implies that there is a surjective function from the subset of $\partition{\Xi}$ to the one of $\partition{\Lambda}$, which means that $|\partition{\Lambda}| \leq |\partition{\Xi}|$.
         Now $|\partition{\Lambda}| = |\partition{\Xi}|$ iff this function is a bijection and the cells of the two partitions are equal one by one. It follows that  whenever $\partition{\Lambda} \succeq \partition{\Xi}$  and  $|\partition{\Lambda}| \neq |\partition{\Xi}|$, $\partition{\Lambda}$ must have fewer cells.
    \end{proof}

    Given a permutation $\sigma$ in the symmetric group $\symg{n}$, we let it act naturally on a partition such that, given a partition $\partition{\Lambda}$ and one of its cells $\Lambda_i = \{a_1, a_2, \dots a_j\}$:

    \begin{equation}
        \sigma \circ \Lambda_i = \{ \sigma (a_1), \sigma (a_2), \dots, \sigma (a_j) \}.
    \end{equation}

    Disjoint subsets will give disjoint images so that the action can be extended to the full partition:
    \begin{equation}
        \sigma \circ \partition{\Lambda} = \{\sigma \circ \Lambda_1, \dots , \sigma \circ \Lambda_k\}.
    \end{equation}
    It is also possible to assign a partition to each permutations in $\symg{n}$, which is naturally induced by the orbit of its cycles $\sigma_i$. We denote such a partition as $\partition{\sigma}$.
    The order relation defined earlier allows to confront partitions and permutations and is useful to describe whether a permutation is a symmetry for a partition, that is if $\sigma \circ \partition{\Lambda} = \partition{\Lambda}$. This conditions is met if and only if the permutations move the elements of $\Omega$ within the cells of $\partition{\Lambda}$ and not between them. This is equivalent to saying that each cycle of $\sigma$ is always contained in a cell of $\partition{\Lambda}$, that is $\partition{\Lambda} \succeq \partition{\sigma}$.

    The group of symmetries of a partition $\sympart{\Lambda}$ can then be swiftly described as:
    \begin{equation}
        \sympart{\Lambda} = \{\sigma \in \symg{n} \mid \partition{\sigma} \preceq \partition{\Lambda}\}.
    \end{equation}

    Using this notation we now give a independent proof of the invertibility of the system of equations \ref{eq:msigma_sum1}.

    \begin{proposition}[Permutation to partition matrix]
        \label{lm:perm_to_part}
            Let $R_{\partition{\sigma}\partition{\Lambda}}$ be the $\belln{n} \times \belln{n}$ matrix which describes the linear relation between the generalized indistinguishabilities and the partition distribution, such that $M_\sigma = \sum_{\partition{\Lambda}} R_{\partition{\sigma}\partition{\Lambda}} \wpart{\Lambda}$. Then $A_{\partition{\sigma}\partition{\Lambda}}$ takes the form:
    
            \begin{equation}
            \label{eq:perm_to_part}
                R_{\partition{\sigma}\partition{\Lambda}} =
                \begin{cases}
                1 & \text{if} \; \partition{\Lambda} \succeq \partition{\sigma}\\
                0 & \text{otherwise}
                \end{cases}.
            \end{equation}
            In particular $R_{\partition{\sigma}\partition{\Lambda}}$ is invertible.
        \end{proposition}
        \begin{proof}
            The form of the matrix derives from direct inspection of \refeqn{eq:msigma_sum1} and so it remains to prove that the matrix is invertible.
            
            Lemma~\ref{lm:deferrec_celln} shows that if $\partition{\Lambda} \succeq \partition{\sigma}$ then either $\partition{\Lambda} = \partition{\sigma}$ or $|\partition{\Lambda}| \leq |\partition{\sigma}|$.
            It follows that by sorting the permutation (rows) and partitions (columns) by decreasing number of cells, $R_{\partition{\sigma}\partition{\Lambda}}$ becomes a lower-triangular matrix.
            Since $\partition{\Lambda} \succeq \partition{\Lambda}$, the diagonal elements $R_{\partition{\Lambda}\partition{\Lambda}}$ are all equal to 1. It directly follows that $\mathrm{Det}\!\left( R_{\partition{\sigma}\partition{\Lambda}} \right) = 1$ and $R_{\partition{\sigma}\partition{\Lambda}}$ is invertible
        \end{proof}
We highlight how the structure of the linear system of Prop.~\ref{lm:perm_to_part} that allowed to prove invertibility, also suggests an algorithm to extract the partition distribution. 
It is in fact known that lower-triangular matrices can be efficiently inverted via forward substitution in quadratic time.
The matrix elements $R^{-1}_{\partition{\Lambda}\partition{\sigma}}$ underlying the inversion procedure correspond to the M\"obius function $\mu_{\partition{\Lambda}\partition{\sigma}}$ \cite{rota_foundations_1964} that realizes the inverse relation:

\begin{equation}
    \wpart{\Lambda} = \sum_{\partition{\sigma}} \mu_{\partition{\Lambda}\partition{\sigma}} M_{\partition{\sigma}}
\end{equation}

    We add one last lemma on conjugation, required for Sec.~\ref{ssec:twirling}:
    \begin{lemma}
    \label{lm:connect_conj}
        Given two permutations $\sigma$ and $\tau \in \symg{n}$ such that $\partition{\sigma} = \partition{\tau}$, there always exists a third permutation $\nu$ such that $\tau = \nu \sigma \nu^{-1}$.
    \end{lemma}
    \begin{proof}
        Since $\sigma$ and $\tau$ have the same cycle structure, one can write the disjoint cycles of the two permutations and their respective elements side by side such as:
        \begin{equation*}
        \begin{split}
            (s_0,s_1,s_2)(s_3, s_4, \cdots)\cdots(\cdots s_n),\\
            (t_0,t_1,t_2)(t_3, t_4, \cdots)\cdots(\cdots t_n).
        \end{split}
        \end{equation*}

        Using the two orderings $s_i$ and $t_i$ of the set $\{1, \cdots,n\}$. One can then construct $\nu$ by setting $\nu(s_i) = t_i$.
    \end{proof}

%%%%%%%%%%%%%%%%%%%%%%%%%%%%%%%%%%%%%%%%%%%%%%%%%%%%%%%%%%%%%%%%%%%%%%%%%%%%%%%%

\section{Partition states}
\label{app:partstate}

It is possible to build a set of representative states $\{\partket{\Lambda}\}$ in one to one relation with the set $\belln{n}$, making use of the following the membership function $s_{\partition{\Lambda}}$:

\begin{equation}
    s_{\partition{\Lambda}} : \{1,...,n\} \rightarrow \{1,...,|\partition{\Lambda}|\} \quad | \quad s_{\partition{\Lambda}}(i) = j \iff i \in \Lambda_j.
\end{equation}

Arbitrarily choosing a single-photon basis $\{\ket{1_k}\}$, each partition state takes the following form:

\begin{equation}
    \partket{\Lambda} = \bigotimes_{j=0}^n \ket{1_{s_{\partition{\Lambda}} (i)}}_i.
\end{equation}

Since two distinct partitions have membership functions that differ by at least one value, it is clear that for $\partition{\Lambda},\partition{\Xi} \in \belln{n}$ we have $\scalar{\psi_{\partition{\Lambda}}}{\psi_{\partition{\Xi}}}= \delta_{\partition{\Lambda} \partition{\Xi}}$.

The interesting property of the partition states is that they behave well under the mode permutation operators of equation (\ref{eq:evolution_permutation}).
In appendix~\ref{app:partition}, we showed that one can define an action of the permutations $\sigma \in \symg{n}$ on any partition $\partition{\Lambda} \in \belln{n}$. This action translates directly to the effect of permutation operators on the partition states:

\begin{equation}
    \sigma \circ \partition{\Lambda} = \partition{\Xi} \follows \opperm{\sigma}\partket{\Lambda} = \partket{\Xi}.
\end{equation}

Such action also induces a stabilizer subgroup $\sympart{\Lambda}$, which represents the symmetries of a partition as well as the one of the relative partition state. The cosets relative to these symmetries, $\tau \symg{\partition{\Lambda}}$ with representatives  $[\tau]$, list all the distinct partitions $\tau \circ \partition{\Lambda}$ that can be obtained from $\partition{\Lambda}$ through mode permutation, where now the ordering of subsets is relevant. 

Evolving $\partket{\Lambda}$ according to \refeqn{eq:evolution_permutation}, and grouping up all elements in the right-hand side that result in the same partition, we obtain:

\begin{equation}
    \hat{O}^{\outcome}\mathcal{U}\partket{\Lambda} =\sum_{[\tau]} \left( \sum_{\egpartperm{\sigma}{\Lambda}} X_{\tau\sigma} \right) \ket{\psi_{\tau \circ \partition{\Lambda}}},
\end{equation}
from which one can compute the outcome probability as the norm of the state vector:

\begin{equation}
\label{eq:part_path_interference}
    p(\outcome | \partition{\Lambda}) = \sum_{[\tau]} \modsq{ \sum_{\egpartperm{\sigma}{\Lambda}} X_{\tau\sigma} }.
\end{equation}

The path interference structure of the above formula can be further confirmed with what is obtained by computing Shchesnovich's generalized indistinguishabilities:

\begin{equation}
\label{eq:part_msigma}
    M_\sigma = \langle \opperm{\sigma} \rangle =
    \begin{cases}
       1  & \text{ if } \sigma \in \symg{\partition{\Lambda}}\\
       0  & otherwise \\
    \end{cases}.
\end{equation}

The well-structured incoherence of the distinguishable cells of a partition state can be made more explicit by expressing the sum of path coefficients $X_\sigma$ (see \refeqn{eq:path})  in terms of sub-permanents of the scattering matrix $U^{\outcome}$.
We note that the group of symmetries of a partition $\sympart{\Lambda}$ is actually the product group of the symmetries of each one of its subsets, which are in turn isomorphic to a symmetric group of size $|\Lambda_i|$:

    \begin{equation}
    \label{eq:symg_sep}
        \sympart{\Lambda} = \bigtimes_i \symg{\Lambda_i},
    \end{equation}
$\sigma$ can thus always be written as  $\sigma_1\sigma_2...\sigma_k$, where $\sigma_i \in \symg{\Lambda_i}$ and the sum over the symmetries can then be broken down into $|\partition{\Lambda}|$ sums over the symmetries of the subsets.
Fixing $[\tau]$ in \refeqn{eq:part_path_interference}, the partial sum inside each modulus can be broken down into the single-photon transition amplitudes and sorted according to the cell they move:

    \begin{equation}
    \label{eq:partition_convolution}
    \begin{split}
       \sum_{\egpartperm{\sigma}{\Lambda}} X_{\tau\sigma}   
       =& \sum_{\sigma_1 \in \symg{\Lambda_1}} \sum_{\sigma_2 \in \symg{\Lambda_2}}...\sum_{\sigma_k \in \symg{\Lambda_k}} \prod_i U_{i \:\tau(\sigma(i))} \\
       =& \left(\sum_{\sigma_1 \in \symg{\Lambda_1}} \prod_{i \in \Lambda_i} U_{i \:\tau(\sigma_1(i))}\right)\left(\sum_{\sigma_2 \in \symg{\Lambda_2}} \prod_{i \in \Lambda_1} U_{i \:\tau(\sigma_2(i))}\right)...\left(\sum_{\sigma_k \in \symg{\Lambda_k}} \prod_{i \in \Lambda_k} U_{i \:\tau(\sigma_k(i))}\right).
    \end{split}
    \end{equation}

    The terms inside each pair of parentheses can be seen as the sub-permanents of the outcome matrix relative to the photons in the cell $\Lambda_i$ coherently evolving into the modes $[\tau] \circ \Lambda_i$.
    The overall photocounting statistics is then a result of a classical convolution between the statistics relative to each distinguishable cell: the incoherent sum over $[\tau]$ lists all the configurations in which the cells $\Lambda_i$ can independently evolve to form the outcome $\outcome$.

\section{Examples of coherent and incoherent distinguishability}
\label{app:examples}

    Knowing that orbit invariance is the criterion for a partition representation to exist, it becomes straightforward to construct a state for which such a representation is not possible.
    Since $\sigma$ and $\sigma^{-1}$ have the same orbits, and $M_\sigma = M_{\sigma^{-1}}^*$, orbit invariance immediately implies that $M_\sigma \in \mathbb{R}$ for all $\sigma$.
    Therefore, all collective phases \cite{shchesnovich_collective_2018} are constrained to be equivalent to either $0$ or $\pi$.
    The smallest number of modes for which this is not automatically true is three.
    Considering a single-photon orthogonal basis $\{\ket{1_a},\ket{1_b}, \hdots \}$ one can construct the following three-photon state.
    
    \begin{example}[Triad phase \cite{menssen_distinguishability_2017}]
        \label{ex:triad_phase}
        The state:
        $$ \ket{\psi} = \ket{1_a} \otimes \frac{1}{\sqrt{2}} (\ket{1_a} + \ket{1_b})\otimes \frac{1}{\sqrt{2}} (\ket{1_a} + e^{i \phi} \ket{1_b})$$

        \noindent
        has $M_{(1 2 3)} = \frac{1}{4}(1+ e^{i \phi})$ and thus does not have a partition representation unless $\phi =m\pi$ for $m$ integer.
        \end{example}

    Provided that a partition representation can be established, the set of partition weights sum to one. This can be shown by considering \refeqn{eq:msigma_sum1} for the identity permutation:
    
    \begin{equation}
        \sum_{\partition{\Lambda}} \wpart{\Lambda} = M_{\mathrm{id}} = \trace{\rho}.
    \end{equation}

    \noindent
    However, when attempting to compute generalized indistinguishabilities of a state with a triad phase being equal to $\phi=(2m+1)\pi$, it becomes clear that the set $\{\wpart{\Lambda}\}$ is actually a quasi-probability distribution, with no constraint on being positive:

    \begin{example}[Negative partition distribution]
        \label{ex:negative_partition}
    
            The pure state
            $$ \ket{\psi} = \frac{1}{2\sqrt2}\left( \ket{1_c}+ \ket{1_a}\right)\otimes \left( \ket{1_a}+ \ket{1_b}\right)\otimes \left( -\ket{1_b}+ \ket{1_c}\right)$$

            \noindent
            has the following six generalized indistinguishabilities:
    
            \begin{equation*}
            \begin{split}
                M_{(123)} = M_{(132)}  = -\frac{1}{8}\\
                M_{(12)} = M_{(13)} =  M_{(32)}  = \frac{1}{4}\\
                M_{\mathrm{id}} = 1.            \end{split}
            \end{equation*}
            This corresponds to the partition distribution:
            \begin{equation*}
            \begin{split}
                p_{\{1,2,3\}} = -\frac{1}{8}\\
                p_{\{1,2\}\{3\}} = p_{\{1,3\}\{2\}} = p_{\{2,3\}\{1\}} = \frac{3}{8}\\
                p_{\{1\}\{2\}\{3\}} = 0,
            \end{split}
            \end{equation*}
            which sums to 1 while containing a negative weight.
        \end{example}

    multi-photon states that are simultaneously diagonalizable are instead the original example of states with a proper positive partition distribution, and this class includes the common scenario dealing with $n$ copies of the same single-photon state $\rho_0$.
    In this case, assuming $\rho_0$ has a discrete eigenspectrum $\{ p_\lambda\}$ of size $|\{ p_\lambda\}|$, the partition distribution is a linear function of the generalized means $\mu_k = \left(\frac{\sum_\lambda p_\lambda^k}{|\{p_\lambda\}|}\right)^{\frac{1}{k}}$ of its eigenspectrum.
    For a $k$-cycle $(1~2~\cdots~k)$ the relative generalized indistinguishability in fact reads:
    
    \begin{equation}
        M_{(1~2~\cdots~k)} = \trace{\rho_0^{ k}} =  |\{p_\lambda\}| \cdot \mu_k^k.
    \end{equation}
    According to \refeqn{eq:msigma_separable}, these values are then sufficient to obtain all the other indistinguishabilities.
    A state of the form $\rho_0^{\otimes n}$ is the most relevant for the analysis of a Boson Sampling implementation that relies on just one periodic single-photon source and time-to-space demultiplexing to produce the desired multi-photon initial state: as a first approximation, the source will always produce the same effective single-photon state.

    A last relevant example is a class of states possessing a partition representation but where the constituent separable components are not simultaneous diagonalizable. An insightful example of such a state is found from the Orthogonal Bad Bits (OBB) distinguishability model \cite{sparrow_quantum_2017,renema_classical_2019}.

    \begin{example}[OBB]
    \label{ex:obb}
        There exist a family of separable $n$-photon states $\rho = \bigotimes_i^n \rho_i $ which possess a partition representation, but such that $[\rho_i, \rho_j] \neq 0$ for all $i\neq j$.

        \noindent The state is obtained by setting:

        $$ \rho_i =  (\sqrt{x}\ket{\underline{1}} + \sqrt{1 - x}\ket{1_i})(\sqrt{x}\bra{\underline{1}} + \sqrt{1 - x}\bra{1_i}), $$
        
        \noindent where  $x \in [0,1]$,  $\ket{\underline{1}}$ is a single photon in a mode common to all $\rho_i$, and $\ket{1_i}$ is a single photon in a mode orthogonal to all other modes:
        $$\scalar{1_i}{\underline{1}}=0, \scalar{1_i}{1_j}=\delta_{ij}.$$
    \end{example}

    In this model, a single parameter $x$ governs the amount of distinguishability by tuning the share of the common mode $\ket{\underline{1}}$ in each single-photon state, driving a transition between a perfectly indistinguishable statistics when $x=1$ and perfectly distinguishable statistics when $x=0$.
    The generalized indistinguishabilities can be computed with Prop.~\ref{eq:msigma_separable} and are orbit invariant, depending only on the number of fixed points of the permutation:

    \begin{equation}
    \label{obb_genind}
        M_\sigma = \prod_{i \neq \sigma(i)} |x|.
    \end{equation}

    The partition distribution of an OBB state is also easily recovered: although each $\rho_i$ is pure, the coherence between $\ket{\underline{1}}$ and  $\ket{1_i}$ is photocounting-irrelevant. The partially distinguishable photon $\rho_i$ is thus equivalent to a classical mixture of the desired `signal' photon $\ket{\underline{1}}\bra{\underline{1}}$ with probability $x$ and the undesired `noise' photon $\ket{1_i}\bra{1_i}$ with probability $1-x$.
    One then finds $\sum_{k=0}^n\binom{n}{k}=2^n$ distinct partition configurations: each one of the $\binom{n}{k}$ states with $k$ indistinguishable and $n-k$ perfectly distinguishable photons is found with a coin-flip probability of $x^k (1-x)^{n-k}$. We also highlight how this number is smaller than the expected Bell number $\belln{n}$ for a general partition representation.

%%%%%%%%%%%%%%%%%%%%%%%%%%%%%%%%%%%%%%%%%%%%%%%%%%%%%%%%%%%%%%%%%%%%%%%%%%%%%%%%
\section{$L_2$ distance for twirled state}
\label{app:l2dist}
Upon grouping up the terms which depend on the same $M_\sigma$, Eq.~(\ref{eq:path_interf}) shows how we can represent the probability of a photocounting outcome as a linear function of the indistinguishabilities $M_\sigma$ (\refeqn{eq:path_interf_group} of App.~\ref{app:cyclic}), with the ideal indistinguishable distribution obtained by setting $M_\sigma = 1~\forall\sigma$.
We can therefore consider a random outcome $\outcome$ and express the observed difference in probabilities as:

\begin{equation}
    \Delta p(\outcome) = p_0(\outcome) - p(\outcome) =\sum_{\sigma \neq id} (1 -M_\sigma) \sum_{\tau}X^*_{\tau}X_{\tau\sigma}.
\end{equation}

We shall use $q_\sigma$ to denote $\sum_{\tau}X_{\tau\sigma} X^*_{\tau}$ and rely on some useful results detailed in Ref.~\cite{hoven_efficient_2024} on the average behavior over truncated Haar-random matrices (in their notation $ q_\sigma = \permament{M\circ M^*_{\sigma^{-1}}}$, where $M$ would be the scattering matrix relative to the outcome $\outcome$).
In Ref.~\cite{hoven_efficient_2024}, it is shown that $q_\sigma$ behave like independent random variables with zero mean and a variance that depends only on the number $j$ of fixed points of $\sigma$. For $\sigma,\tau \neq id$ we thus have:

\begin{equation}
    \avg{q_\sigma} = 0,
\end{equation}
\begin{equation}
    \avg{q_\sigma q_\tau} = \delta_{\sigma\tau} \cdot \frac{n!}{m^{2n}}\sum_{p=0}^{n-j} R_{n-j,p}2^p = \var{q_\sigma},
\end{equation}
\begin{equation}
    \avg{\Delta p} =0,
\end{equation}
\begin{equation}
    \var{\Delta p} =\avg{\Delta p^2},
\end{equation}
where $R_{n,j}$ is the rencontres number, which counts the permutation of $n$ objects with exactly $j$ fixed points.
Using the function $\mathrm{fix}(\sigma)$ to count the number of fixed points of $\sigma$ and grouping up the terms with the same number of fixed points one obtains:
\begin{equation}
        \Delta p = \sum_{j = 0}^{n-1} \sum_{\mathrm{fix}(\sigma) = j} (1 -M_\sigma) q_\sigma.
\end{equation}

We shall now focus on a specific inner sum, which due to the dependence of $q_\sigma$ on the number of fixed points contains i.i.d random variables $q^{(j)}$.
We can further count the number of terms in each inner sum through the rencontres number.
The variance of their sum reads:
\begin{equation}
   \var{\sum_{\mathrm{fix}(\sigma) = j} (1 -M_\sigma) q_\sigma} = \var{q^{(j)}} \cdot R_{n,j} \left[(1 -\avg{M_\sigma})^2 + \var{M_\sigma} \right],
\end{equation}
where the moments $\avg{M_\sigma}$, $\avg{M_\sigma^2}$ are taken over the set $\mathrm{fix}(\sigma) = j$.
Any averaging between the $M_\sigma$ will leave $\avg{M_\sigma}$ untouched while decreasing $\var{M_\sigma}$, decreasing the overall variance.
Extending this reasoning to all $j$, denoting the moments of the indistinguishabilities by $\avg{\cdot}_j, \var{\cdot}_j$, with $j$ fixed points we obtain:

\begin{equation}
    \var{\Delta p} = \sum_{j=0}^{n-1}\var{q^{(j)}} \cdot R_{n,j} \left[(1 -\avg{M_\sigma}_j)^2 + \var{M_\sigma}_j \right].
\end{equation}

We thus conclude that any averaging operation between indistinguishabilities with the same number of fixed points will result in particle statistics that lies closer to the ideal distribution, on average.
We can finally observe that both the physical permutation twirling operation (Def.~\ref{def:proj_twirl}) and the possibly unphysical strict partition projection (Def.~\ref{def:proj_strict}) respect this requirement.

%%%%%%%%%%%%%%%%%%%%%%%%%%%%%%%%%%%%%%%%%%%%%%%%%%%%%%%%%%%%%%%%%%%%%%%%%%%%%%%%
\section{Symmetric subspace projection for partition states}
\label{app:sym_proj}

Given a partition state $\partket{\Lambda}$ and the projector onto the symmetric subspace $\Pi_{\mathrm{sym}} = \frac{1}{n!} \sum_\sigma \opperm{\sigma}$, the modulus of the projection of the former onto the latter can be computed with a closed formula:

\begin{equation}
    \partbra{\Lambda} \Pi_{\mathrm{sym}} \partket{\Lambda} = \frac{1}{n!} \sum_\sigma \partbra{\Lambda} \opperm{\sigma} \partket{\Lambda} = \frac{|\sympart{\Lambda}|}{n!}. 
\end{equation}

\noindent These stabilizer groups can be broken down into the product of smaller permutation groups \refeqn{eq:symg_sep}) and their size can be computed in terms of the size of the cells of $\partition{\Lambda}$:

\begin{equation}
    |\sympart{\Lambda}| = \prod_i |\Lambda_i|!.
\end{equation}

Given the inequality $(n-k)! \, k! \leq (n-1)!$ for $k\leq n$, which is true whenever $k \neq 0,n$, we deduce that given a two-set partition $\partition{\Lambda} = \{ \Lambda_1, \Lambda_2 \}$, the size of $\sympart{\Lambda}$ is always less than the one relative to the one-subset partition $\partition{\Xi} = \{\Lambda_1 \cup \Lambda_2\}$.
By repeated application of this dividing argument, one can bound the size of the stabilizer of a partition $\partition{\Lambda}$ by the one relative to any $\partition{\Xi} \succeq \partition{\Lambda}$.
In particular, if $\partition{\Lambda} \neq \{1,2,\dots,n\}$ then $|\sympart{\Lambda}| \leq (n-1)!$ and the following bound is always true:
\begin{equation}
    \frac{|\sympart{\Lambda}|}{n!} \leq \frac{1}{n}.
\end{equation}

\end{document}